%% file: agnostic-main.tex
\documentclass[11pt]{article}

\usepackage{amsmath,amssymb,amsfonts,amsthm,epsfig}
\usepackage[usenames,dvipsnames]{xcolor}
\usepackage{algorithm}
\usepackage{algorithmicx}
\usepackage{algpseudocode}
\usepackage{bm,xspace}
\usepackage{tcolorbox}
\usepackage{cancel}
\usepackage{fullpage}
\usepackage{liyang}
\usepackage{framed}
\usepackage{verbatim}
\usepackage{enumitem}
\usepackage{array}
\usepackage{multirow}
\usepackage{afterpage}
\usepackage{mathrsfs}
\usepackage{pifont} 
\usepackage{chngpage}
\usepackage[normalem]{ulem}
\usepackage{boxedminipage}
\usepackage{caption}
\usepackage{thmtools}
\usepackage{thm-restate}
\usepackage{float}
\usepackage{amssymb}
\usepackage{booktabs} 
\usepackage[dvipsnames]{xcolor}

\newcommand{\thickhline}{%
    \noalign {\ifnum 0=`}\fi \hrule height 1pt
    \futurelet \reserved@a \@xhline
}

\usepackage[dvipsnames]{xcolor}

\def\colorful{1}

\ifnum\colorful=1

\newcommand{\red}[1]{{\color{red} {#1}}}

\fi
\ifnum\colorful=0

\newcommand{\red}[1]{{{#1}}}

\fi

\newcommand{\norm}[1]{\left\lVert#1\right\rVert}
\newcommand{\vect}[1]{\mathbf{#1}}
\newcommand{\abs}[1]{\left\lvert#1\right\rvert}
\newcommand{\fig}[1]{\left\{#1\right\}}

\DeclareMathOperator*{\argmax}{arg\,max}
\DeclareMathOperator*{\argmin}{arg\,min}

\newcommand{\viol}{\mathrm{viol}}

\usepackage{epsfig}
\usepackage{color}
\usepackage{times}

\newcommand{\spn}[1]{\text{Span}\left( #1 \right)}
\DeclareMathOperator*{\argmaxop}{arg\,max}
\DeclareMathOperator*{\argminop}{arg\,min}
\usepackage{bbm}

\newcommand{\myeqold}[1]{\stackrel{\mathclap{\normalfont\mbox{#1}}}{=}}
\newcommand{\myleqold}[1]{\stackrel{\mathclap{\normalfont\mbox{#1}}}{\leq}}
\newcommand{\mygeqold}[1]{\stackrel{\mathclap{\normalfont\mbox{#1}}}{\geq}}

\begin{document}
\global\long\def\red#1{\text{{\color{red}#1}}}%

\global\long\def\dimvc{\dim_{\text{VC}}}%

\global\long\def\spn#1{\text{Span}\left(#1\right)}%

\global\long\def\argmax{\argmaxop}%

\global\long\def\argmin{\argminop}%

\global\long\def\E{\mathbb{E}}%

\global\long\def\R{\mathbb{R}}%

\global\long\def\P{\mathbb{P}}%

\global\long\def\C{\mathbb{C}}%

\global\long\def\Z{\mathbb{Z}}%

\global\long\def\F{\mathcal{F}}%

\global\long\def\B{\mathcal{B}}%

\global\long\def\Ep{\mathcal{E}}%

\global\long\def\sign{\mathbb{\text{sign}}}%

\global\long\def\myeq#1{\myeqold{#1}}%

\global\long\def\myleq#1{\myleqold{#1}}%

\global\long\def\mygeq#1{\mygeqold{#1}}%

\global\long\def\indicator{\mathbbm{1}}%

\global\long\def\parr#1{\left(#1\right)}%

\global\long\def\pars#1{\left[#1\right]}%

\global\long\def\fig#1{\left\{  #1\right\}  }%

\global\long\def\abs#1{\left|#1\right|}%

\global\long\def\norm#1{\left\Vert #1\right\Vert }%

\global\long\def\bra#1{\left\langle #1\right\rangle }%

\global\long\def\d{\,d}%

\global\long\def\conditional{\bigg\vert}%

\global\long\def\vect#1{\bm{#1}}%

\global\long\def\mydot#1{\stackrel{\mathbf{.}}{#1}}%

\global\long\def\mydoubledot#1{\stackrel{\mathbf{.}\mathbf{.}}{#1}}%

\global\long\def\mytripledot#1{\stackrel{\mathbf{.}\mathbf{.}\mathbf{.}}{#1}}%

\global\long\def\mybar#1{\overline{#1}}%

\global\long\def\myhat#1{\widehat{#1}}%

\global\long\def\spn#1{\mathop{\text{span}}\left(#1\right)}%

\global\long\def\conv#1{\mathop{\text{conv}}\left(#1\right)}%

\global\long\def\righ{\rightarrow}%

\global\long\def\var{\text{var}}%

\global\long\def\N{\mathcal{N}}%

\title{Agnostic proper learning of monotone functions: beyond the black-box correction barrier}
\date{}
		\author{
  Jane Lange\thanks{MIT, {\tt jlange@mit.edu}. Supported in part by NSF Graduate Research Fellowship under Grant No. 2141064, NSF Award CCF-2006664, DMS-2022448, Big George Fellowship, Akamai Presidential Fellowship and Google. }
			\and
   Arsen Vasilyan \thanks{MIT, {\tt vasilyan@mit.edu}.
			Supported in part by NSF awards CCF-2006664, DMS-2022448, CCF-1565235, CCF-1955217, Big George Fellowship and Fintech@CSAIL.}
		}
\maketitle

\input{introduction}

\input{preliminaries}

\input{pseudocode}
\input{l1-corrector}

\input{TheEllipsoidThing.tex}
\section{Acknowledgments}
We thank Ronitt Rubinfeld and Mohsen Ghaffari for helpful conversations about local computation algorithms. We additionally thank Ronitt Rubinfeld for useful comments regarding the manuscript.

\bibliographystyle{alpha}
\bibliography{notesNew, references}

\appendix
\input{appendix}

\end{document}

%% file: introduction.tex
\input{abstract}

\section{Introduction} 

The class of monotone functions over $\{\pm 1\}^n$ is an object of major interest in theoretical computer science. In consequence, the study of learning and testing algorithms for monotone functions 
\cite{bshouty_fourier_1996, kearns_cryptographic_1989,goldreich_property_1998,blum_learning_1998,dodis_improved_1999,amano_learning_2006, ailon_estimating_2007, odonnell_kkl_2009, chakrabarty_optimal_2013,chen_new_2014,chen_testing_2019,khot_monotonicity_2015,belovs_polynomial_2016,chen_beyond_2017,chakrabarty_adaptive_2019, pallavoor_approximating_2022, lange_properly_2022}  and various subclasses of monotone functions \cite{angluin_queries_1988,jackson_learning_2011,yang_learnability_2013,blanc_properly_2022} is a major research direction.
In this work, we consider two  fundamental problems in this line of work: \emph{approximating the distance} of unknown functions to monotone, 
and \emph{agnostic proper learning} of monotone functions. 
For each of these problems we are given independent uniform samples $\{\vect{x}_i\}$ labeled by an arbitrary function $f: \{\pm 1\} \righ \{\pm 1\}$ and we are required to perform the following tasks:
\begin{enumerate}
    \item \textbf{Estimating distance to monotonicity} is the task of estimating up to some additive error $\epsilon$ the distance $\dist(f, f_{\text{mon}})$ from $f$ to the monotone function $f_{\text{mon}}$ that is closest to $f$.
    \item \textbf{Agnostic proper learning of monotone functions} is the task of
    obtaining a description of a monotone function $g_{\text{mon}}$, whose distance $\dist(f, g_{\text{mon}})$ approximates $\dist(f, f_{\text{mon}})$ up to additive error $\epsilon$. 
\end{enumerate}
Prior to our work, it was known that information-theoretically these tasks can be solved using only $2^{\tilde{O}\parr{\sqrt{n}/\eps}}$ samples. However, all known algorithms had a run-time of $2^{\Omega(n)}$, thus dramatically exceeding the known sample complexity of $2^{\tilde{O}\parr{\sqrt{n}/\eps}}$. 
In this work, we close this gap in our knowledge and give algorithms for the two tasks above that not only use $2^{\tilde{O}\parr{\sqrt{n}/\eps}}$ samples, but also run in time $2^{\tilde{O}\parr{\sqrt{n}/\eps}}$. 
This nearly matches the $2^{\tilde{\Omega}(\sqrt{n})}$ lower bound of \cite{blais_learning_2015}.

\subsection{Previous work}

We note that the work of \cite{lange_properly_2022} largely concerns itself with the problem of \emph{realizable learning} of monotone functions, i.e. learning a function $f$ that is itself promised to be monotone. 
In contrast, the focus of our work is the harder setting when the function $f$ we access is \emph{arbitrary} and we want to obtain a description of a monotone function $g_{\text{mon}}$ that predicts $f$ best among monotone functions (up to an additive slack of $\epsilon$). 

Still, as noted in \cite{lange_properly_2022}, their work does give mixed additive-multiplicative approximation guarantees in the settings we study here. Specifically, \cite{lange_properly_2022} gives algorithms that also run in time $2^{\tilde{O}\parr{\sqrt{n}/\eps}}$ and achieve the following:
\begin{enumerate}
    \item Obtain a $(3, \epsilon)$-approximation of $\dist(f, f_{\text{mon}})$. In other words, the estimate is in the interval between $\dist(f, f_{\text{mon}})$ and $3 \cdot \dist(f, f_{\text{mon}})+\epsilon$. (We also note that \cite{lange_properly_2022} additionally present an algorithm that gives a distance estimate in $\pars{\dist(f, f_{\text{mon}}), 2 \cdot \dist(f, f_{\text{mon}})+\epsilon}$ but also requires query access to function $f$).
    \item 
    Obtain a succinct description of a monotone function $g_{\text{mon}}$, whose distance $\dist(f, g_{\text{mon}})$  is a $(3, \epsilon)$-approximation to $\dist(f, f_{\text{mon}})$. In other words, it is in the interval between $\dist(f, f_{\text{mon}})$ and $3 \cdot \dist(f, f_{\text{mon}})+\epsilon$. As it is noted in \cite{lange_properly_2022}, this yields a fully agnostic learning algorithm only if $\dist(f, f_{\text{mon}})\leq O(\epsilon)$.
\end{enumerate}
Overall, \Cref{tbl: comparison of our work and previous work} summarizes how our work compares with what was known previously.

\begin{table}[!htbp]\begin{center}
\begin{tabular}{p{4.5cm}  c  c  c } 
 \hline
 Work & \begin{tabular}{@{}c@{}} Guarantee for distance estimate and \\ error for proper agnostic learning \end{tabular}   & Sample complexity & Run-time    \\ \midrule
 \cite{bshouty_fourier_1996, kalai_agnostically_2005} with refinement from \cite{feldman_tight_2020} & $\pars{\dist(f, f_{\text{mon}}), \dist(f, f_{\text{mon}})+\epsilon}$ & $2^{\tilde{O}\parr{\sqrt{n}/\eps}}$ & $2^{\Omega(n)}$   \\ \midrule
 \cite{lange_properly_2022} & $\pars{\dist(f, f_{\text{mon}}), 3 \cdot \dist(f, f_{\text{mon}})+\epsilon}$   & $2^{\tilde{O}\parr{\sqrt{n}/\eps}}$ & {$2^{\tilde{O}\parr{\sqrt{n}/\eps}}$}  \\ \midrule
  \textbf{This paper} & 
$\pars{\dist(f, f_{\text{mon}}), \dist(f, f_{\text{mon}})+\epsilon}$     & $2^{\tilde{O}\parr{\sqrt{n}/\eps}}$ & $2^{\tilde{O}\parr{\sqrt{n}/\eps}}$ 
         \\ 
 \bottomrule
\end{tabular}
\end{center}

\caption{Comparison of our results to previously known algorithms. }
\label{tbl: comparison of our work and previous work}
\end{table}

\subsection{Main results}
The following are our main results: learning and distance approximation of Boolean functions, and local correction of real-valued functions.
\begin{restatable}{theorem}{mainthm}[Agnostic proper learning of monotone functions\footnote{See \Cref{section: learning with randomized labels} for an extension to functions with randomized labels.}]
\label{theorem: main theorem}There is an algorithm that runs in time
$2^{\tilde{O}\parr{\frac{\sqrt{n}}{\eps}}}$ and, given uniform
sample access to an unknown function $f:\left\{ \pm1\right\} ^{n}\righ\left\{ \pm1\right\} $,
with probability at least $1-\frac{1}{2^{n}}$, outputs a succinct representation of
a monotone function $g:\left\{ \pm1\right\} ^{n}\righ\left\{ \pm1\right\} $
that is $\text{\text{opt}+}O(\epsilon)$-close to $f$, where $\text{opt}$
is the distance from $f$ to the closest monotone function (i.e. the
fraction of elements of $\left\{ \pm1\right\} ^{n}$ on which $f$
and its closest monotone function disagree). 
\end{restatable}

The corollary below follows immediately by the standard method of \cite{parnas_tolerant_2004} that runs the learning algorithm in \Cref{theorem: main theorem} and estimates the distance between $g$ and $f$.
\begin{corollary}[Additive distance-to-monotonicity approximation]
		There is an algorithm with running time and sample complexity $2^{\tilde{O}\parr{\frac{\sqrt{n}}{\eps}}}$ that outputs some estimate $est$ of the distance from $f$ to the closest monotone function $f_{mon}$.
		With probability at least $1 - 2^{-n+1}$, this estimate satisfies the guarantee \[\dist(f, f_{mon}) \le est \le \dist(f, f_{mon}) + O(\eps).\]
\end{corollary}

\begin{restatable}{theorem}{correctorthm}[Local monotonicity correction of real-valued functions]
\label{thm:correction-main}
Let $P$ be a poset with $N$ elements, such that every element has at most $\Delta$ predecessors or successors and the longest directed path has length $h$.
Let $f: P \to [-1,1]$ be $\alpha$-close to monotone in $\ell_1$ distance.
There is an LCA that makes queries to $f$ and outputs queries to $g: P \to [-1,1]$,
such that $g$ is monotone and $||f - g||_1 \le 2\alpha + 3\eps$. 
The LCA makes $(\Delta \log N)^{O(\log h \log(1/\eps))}$ queries, uses a random seed of length $\poly(\Delta \log N)$, and succeeds with probability $1 - N^{-10}$.
\end{restatable}


    \subsection{Our techniques: beyond the black-box correction barrier.}
The algorithms of \cite{lange_properly_2022} follow the following pattern (which we also summarize in \Cref{fig:mesh1}):
\begin{enumerate}
    \item Use \cite{bshouty_fourier_1996, kalai_agnostically_2005, feldman_tight_2020} to obtain a succinct description of a (possibly non-monotone) function $f_{\text{improper}}$ whose distance $\dist(f, f_{\text{improper}})$  is at most $\dist(f, f_{\text{mon}})+\epsilon$. The issue now is that $f_{\text{improper}}$ is not necessarily monotone, and therefore the distance $\dist(f, f_{\text{improper}})$ might dramatically underestimate the true distance to monotonicity $\dist(f, f_{\text{mon}})$. 
    \item Design and use a monotonicity corrector, in order to transform the succinct description of $f_{\text{improper}}$ into a succinct description of some \textbf{monotone} function $g_{\text{mon}}$ that is close to $f_{\text{improper}}$. Formally, \cite{lange_properly_2022} develop a corrector that guarantees that the distance $\dist(f_{\text{improper}}, g_{\text{mon}})$ satisfies
    \begin{equation}
    \label{eq: correction equation}
     \dist(f_{\text{improper}}, g_{\text{mon}})
     \leq c
     \min_{\text{monotone f'}} \dist(f_{\text{improper}}, f') + \epsilon,
    \end{equation}
    where the constant $c$ is $2$. They achieve this by a novel use of \textbf{Local Computation Algorithms} (LCAs) on graphs.
\end{enumerate}
\begin{figure}[h]
    \centering
    \includegraphics[width=1\textwidth]{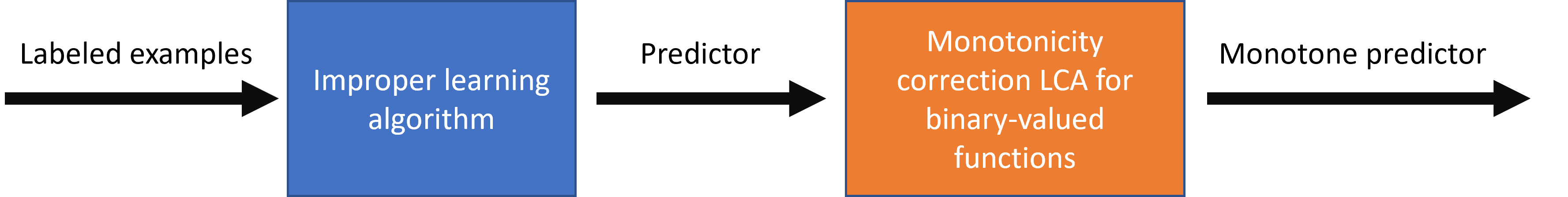}
    \caption{Control-flow diagram of the semiagnostic algorithm of \cite{lange_properly_2022}}
    \label{fig:mesh1}
\end{figure}
This way, \cite{lange_properly_2022} obtain a succinct polytime-evaluable description of a monotone function $g_{\text{mon}}$ for which\footnote{Strictly speaking, the properties of the corrector described so far yield only a guarantee of $\dist(f, g_{\text{mon}})\leq 4 \cdot \dist(f, f_{\text{mon}}) + \epsilon$. To improve the multiplicative error constant from $4$ to $3$ the work of \cite{lange_properly_2022} uses an additional property of the corrector.} $\dist(f, g_{\text{mon}})\leq 3 \cdot \dist(f, f_{\text{mon}}) + \epsilon$. 

However, one can see that even if the correction constant $c$ in \Cref{eq: correction equation} were equal to $1$ (which is the best it can be) this approach could only yield a guarantee of $\dist(f, g_{\text{mon}})\leq 2 \cdot \dist(f, f_{\text{mon}}) + \epsilon$.  

\begin{figure}[h]
    \centering
    \includegraphics[width=1\textwidth]{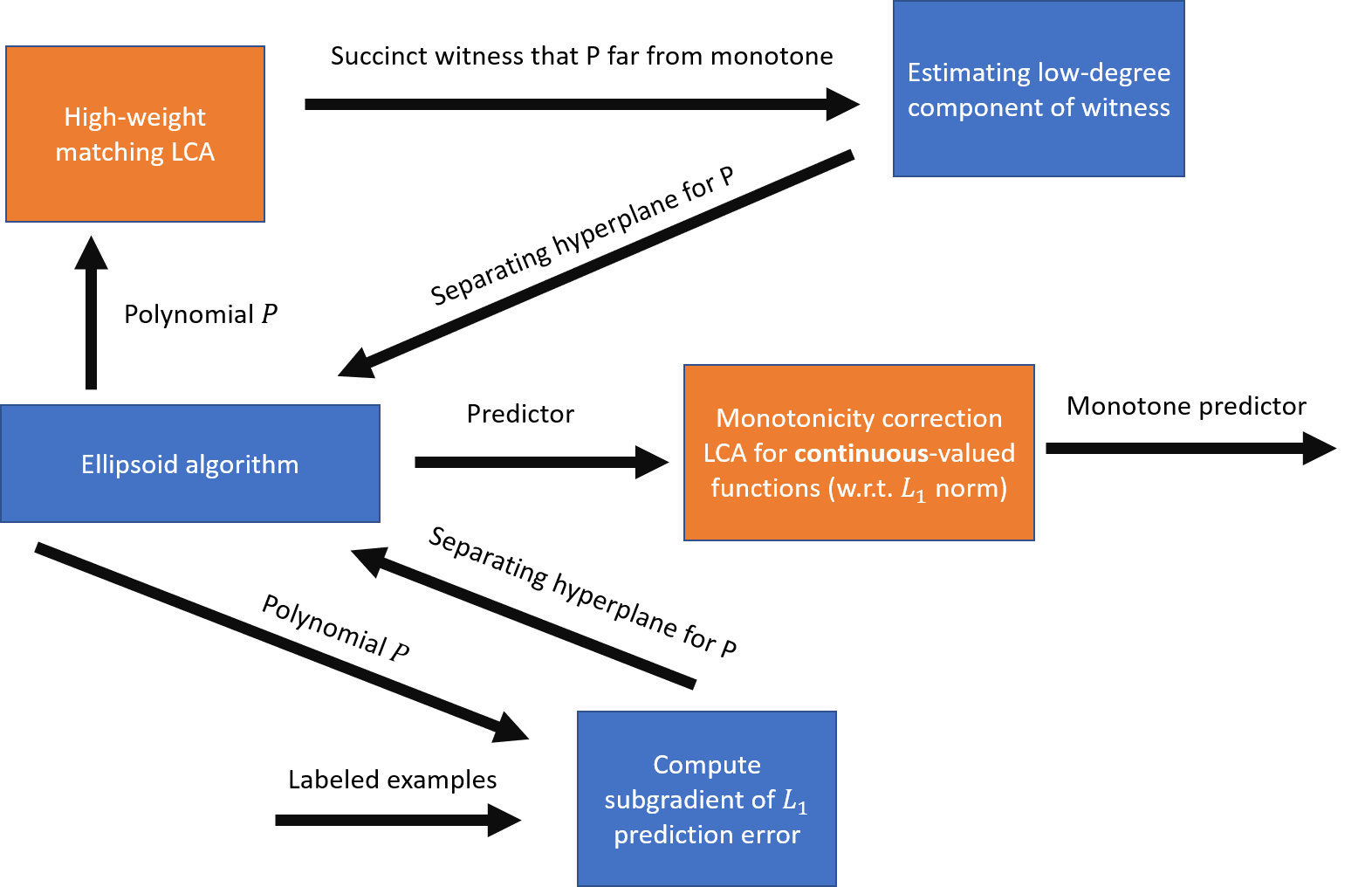}
    \caption{Control-flow diagram of the fully agnostic learning algorithm presented in this work}
    \label{fig:mesh2}
\end{figure}

\subsubsection{Description of our approach}

We overcome this barrier by using a different approach,  summarized in \Cref{fig:mesh2}. As before, there is an improper learning phase and a correction phase; however in both phases we work with real-valued functions. 
We have essentially three steps:

\begin{enumerate}
    \item Find a real-valued polynomial $P$ that is $\epsilon$-close to some monotone function, $(\alpha + \eps)$-close to the unknown function $f$ in $\ell_1$ distance, and bounded in $[-1,1]$. 
    \item Obtain a succinct description of a real-valued function $P_{\text{CORRECTED}}$ that is monotone, and $O(\epsilon)$-close to $P$ in $\ell_1$ distance.
    \item Round the real-valued function $P_{\text{CORRECTED}}$ to be $\{\pm 1\}$-valued, while preserving monotonicity and closeness to $f$. 
\end{enumerate}

In contrast to the approach of \cite{lange_properly_2022}, the improper learning phase is constrained to produce a good predictor that is $\epsilon$-close to some monotone function, regardless of how far $f$ may be from monotone. Existing improper learning algorithms are far from satisfying this new requirement. We design a new improper learner by combining the polynomial-approximation based techniques of \cite{bshouty_fourier_1996, kalai_agnostically_2005, feldman_tight_2020} with graph LCAs and the \emph{ellipsoid method} for convex optimization. 

The improper learning task is a convex feasibility problem; the set of polynomials satisfying the constraints we give in step (1) is a convex subset of the initial convex set of low-degree real polynomials. 
The ellipsoid method requires a \emph{separation oracle}, i.e. some way to efficiently generate a hyperplane separating a given infeasible polynomial from the feasible region.
Such hyperplanes are themselves low-degree real polynomials, which have high inner product with the infeasible polynomial and low inner product with every point in the feasible region.
The separator for the set of polynomials that are $(\alpha + \eps)$-close to $f$ is, as shown in \Cref{fig:mesh2}, just the gradient of the prediction error; 
the more interesting case is the separator for the set of polynomials that are $\eps$-close to monotone. 

With an argument inspired by the characterization of Lipschitz functions given in \cite{berman_l_p-testing_2014}, we observe that if a real-valued polynomial $P$ is far from monotone, this can be witnessed by a large matching on the pairs of 
elements on which $P$ violates monotonicity. 
Given any description of the matching, we show how to extract a separating hyperplane for $P$ by evaluating the matching on a set of sample points.
Therefore, the challenge is to find a description of a sufficiently large matching that can also be evaluated quickly.
We elaborate on this in the next section. 

Step (2) requires another technical contribution, which is an extension of the poset-sorting LCA of \cite{lange_properly_2022} to real-valued functions.
This extension is crucial for us to achieve the overall agnostic learning guarantee, because in the improper learning phase we obtain a real-valued function that is only close to monotone in $\ell_1$ distance.\footnote{One can construct functions that are arbitrarily close to monotone in $\ell_1$ norm but a constant fraction of their values needs to be changed for them to become monotone. Because of this, the corrector of \cite{lange_properly_2022} was not fit for our correction stage.} 
For step (3) we use the rounding procedure of \cite{kalai_agnostically_2005} that rounds real-valued functions to $\{\pm 1\}$-valued functions, and we show that this procedure also preserves monotonicity.

\subsubsection{LCAs and succinct representations of large objects}

In this work we employ heavily the concept of a \emph{succinct representation}. 
The succinct representations we deal with will have size and evaluation time $2^{\tilde{O}(\sqrt{n}/\epsilon)}$.
To be fully specific,
we consider succinct representations of two types of objects:
\begin{itemize}
    \item A succinct representation of a function $f:\{\pm 1\}^n\rightarrow \R$ is an algorithm that, given $ x \in \{\pm 1\}^n$, computes $f(x)$ in time $2^{\tilde{O}(\sqrt{n}/\epsilon)}$. 
    \item A succinct representation of a (possibly weighted) graph $G$ with the vertex set $\{\pm 1\}^n$ is an algorithm that, given $v \in \bits^n$, outputs all its neighbors and the weights of corresponding edges in time $2^{\tilde{O}(\sqrt{n}/\epsilon)}$. 
\end{itemize}
A polynomial of degree $O(\sqrt{n})$ is an example of a succinct representation, but another type of representation that makes frequent appearances in this work is a 
\emph{local computation algorithm}, or LCA \cite{alon_space-efficient_2012, rubinfeld_fast_2011}. 
An LCA efficiently computes a function over a large domain.
For example, an LCA for an independent set takes as input some vertex $v$, makes some lookups to the adjacency list of the graph, then outputs ``yes'' or ``no'' so that the set of vertices for which the LCA would output ``yes'' form an independent set.
Typically, its running time and query complexity are each sublinear in the domain size.  
We require that all LCAs used in this work have outputs consistent with one global object,
regardless of the order of user queries,
and without remembering any history from previous queries.
This property allows us to use the LCA, in conjunction with any succinct representation of the graph,
as a succinct representation of the object it computes. 
We formalize this relationship in \Cref{sec:prelim lca}.

\subsection{Other related work}
The local correction of monotonicity was studied in \cite{ailon_property-preserving_2008,saks_local_2010,bhattacharyya_lower_2010,awasthi_limitations_2014} and \cite{lange_properly_2022} (see \cite{lange_properly_2022} for an overview of previously available algorithms for monotonicity correction and lower bounds).

The work of \cite{canonne_testing_2016} gives an improper learning algorithm for a function class that is larger than monotone functions. Additionally, we note that testing of monotone functions has also been studied over hypergrids \cite{chakrabarty_optimal_2013,berman_l_p-testing_2014,black_od_2018,black_domain_2020}.

In addition to \cite{ghaffari_local_2022}, there have been many exciting recent works on local computation algorithms (LCAs). Some examples include \cite{rubinfeld_fast_2011}, \cite{alon_space-efficient_2012}, \cite{levi_local_2017}, \cite{goos_non-local_2015}, \cite{reingold_new_2016},
\cite{even_best_2014},
\cite{ghaffari_improved_2015}, \\ \cite{chang_complexity_2019}, \cite{even_sublinear_2021}
, \cite{parter_local_2019}, \cite{ghaffari_sparsifying_2019},\cite{levi_local_2020}, \cite{arviv_improved_2021},
\cite{brandt_randomized_2021} and \cite{grebik_classification_2021}.






%% file: abstract.tex
\begin{abstract}
We give the first agnostic, efficient, proper learning algorithm for monotone Boolean functions.
Given $2^{\tilde{O}(\sqrt{n}/\eps)}$ uniformly random examples of an unknown function $f:\{\pm 1\}^n \rightarrow \{\pm 1\}$, 
our algorithm outputs a hypothesis $g:\{\pm 1\}^n \rightarrow \{\pm 1\}$ that is monotone and $(\opt + \eps)$-close to $f$,
where $\opt$ is the distance from $f$ to the closest monotone function.
The running time of the algorithm (and consequently the size and evaluation time of the hypothesis) is also $2^{\tilde{O}(\sqrt{n}/\eps)}$,
nearly matching the lower bound of \cite{blais_learning_2015}. We also give an algorithm for estimating up to additive error $\epsilon$ the distance of an unknown function $f$ to monotone using a run-time of $2^{\tilde{O}(\sqrt{n}/\eps)}$. Previously, for both of these problems, sample-efficient algorithms were known, but these algorithms were not run-time efficient. Our work thus closes this gap in our knowledge between the run-time and sample complexity.

This work builds upon the improper learning algorithm of \cite{bshouty_fourier_1996} and the proper semiagnostic learning algorithm of \cite{lange_properly_2022},
which obtains a non-monotone Boolean-valued hypothesis, then ``corrects'' it to monotone using query-efficient local computation algorithms on graphs.
This black-box correction approach can achieve no error better than $2\opt + \eps$ information-theoretically; 
we bypass this barrier by 
\begin{itemize}
    \item[a)] augmenting the improper learner with a convex optimization step, and
    \item[b)] learning and correcting a real-valued function before rounding its values to Boolean.
\end{itemize} 
Our real-valued correction algorithm solves the ``poset sorting'' problem of \cite{lange_properly_2022} for functions over general posets with non-Boolean labels.
\end{abstract}

%% file: preliminaries.tex
\section{Preliminaries}
\label{sec:prelim}

\subsection{Posets and $\bits^n$}

Let $P$ be a partially-ordered set. 
We use $\preceq$ to denote the ordering relation on $P$.
We say $x\prec y$ (``$x$ is a predecessor of $y$'') if $x\preceq y$ and $x\neq y$,
and use the analogous symbols $\succeq$ and $\succ$ for successorship.
If $x\prec y$ and there is no $z$ in $P$ for which $x\prec z \prec y$,
then $x$ is an \emph{immediate predecessor} of $y$ and $y$ is an \emph{immediate successor} of $x$.
We refer to the poset $P$ and its Hasse diagram (DAG) interchangeably. 
The transitive closure $TC(P)$ is the graph on the elements of $P$ that has an edge from each vertex to each of its successors.
A \emph{succinct representation} of $P$ with size $s$ is any function stored in $s$ bits of memory that takes as input the identity of a vertex,
outputs the sets of immediate predecessors and immediate successors,
and runs in time $O(s)$ in the worst case over vertices.

Specific posets of interest in this work are the Boolean cube and the weight-truncated cube. We give a definition and a size-$O(n/\eps)$ representation computing the truncated cube.
\begin{definition}
The \emph{$n$-dimensional Boolean hypercube} is the set $\{-1,1\}^n$.
For $x,y\in \{-1,1\}^n$, we say $x\preceq y$ if for all $i\in\{1,\cdots,n\}$ one has $x_i\leq y_i$. It is immediate that $\{-1,1\}^n$ is a poset with $2^n$ elements. 

We also define the truncated hypercube 
$$
H^n_\epsilon:=\fig{x \in \{-1,1\}^n:~\abs{\sum_ix_i}\leq \sqrt{2n\log \frac{2}{\epsilon}}},
$$
Via Hoeffding's bound, we have that the fraction of elements in $\{0,1\}^n$ that are not also in $H_n^\epsilon$ is at most $2\exp\left(-\frac{2t^2}{4n}\right)=\epsilon$.
\end{definition}
For $i \in [n]$, let $e_i$ be the vector that is $1$ at index $i$ and $-1$ everywhere else.

\begin{algorithm}[H]
\begin{algorithmic}

\State \textbf{Given:} Input $x \in \bits^n$, truncation parameter $\eps$
\State \Return $\{x \oplus e_i ~|~ i \in [n] \text{ and } |\sum_j (x \oplus e_i)_j| \le \sqrt{2n \log \frac{2}{\epsilon}}\}$

\end{algorithmic}
\caption{\textsc {LCA: TruncatedCube}$(x,\eps)$}
\label{alg:truncated cube}
\end{algorithm}

\subsubsection{Fourier analysis over $\{\pm 1\}^n$.}
Let $[n]$ denote the set $\{1,2,\cdots,n\}$. We define for every $S\subseteq [n]$ the function $\chi_S:\{\pm 1\}^n\righ \R$ as  
$
\chi_S(\vect x):=
\prod_{i \in S} x_i.
$
We define the inner product between two functions $g_1, g_2:\{\pm 1\}^n\righ \R$ as follows: 
$
\langle
g_1,
g_2
\rangle
:=
\E_{\vect x \sim \{\pm 1\}^n} 
\pars{
g_1(\vect x) g_2(\vect x)
}
$.
It is known that $\langle \chi_{S_1}, \chi_{S_2}\rangle =\indicator_{S_1=S_2}$. For a function $g: \{\pm 1\}^n\righ \R$ we denote 
$
\widehat{g}(S):=\langle
g,
\chi_S
\rangle
$.
It is known that 
\begin{align*}
g(\vect x)&=\sum_{S\subseteq [n]} \widehat{g}(S) \chi_S (\vect x)   &  \langle
g_1,
g_2
\rangle
&=
\sum_{S\subseteq [n]} \widehat{g}_1(S) \widehat{g}_2(S).     
\end{align*}

\subsection{Monotone functions}
Part of our algorithm concerns monotonicity of functions over general posets. 
For a function $f: P \to \R$, we say that a pair of elements $x, y \in P$ forms
a \emph{violated pair} if we have $x\preceq y$ but $f(x) >f(y)$,
and we define the \emph{violation score} $\mathrm{vs}(x,y) := f(x) - f(y)$.
The \emph{violation graph } $\viol(f)$ is the subgraph of $TC(P)$ induced by violated pairs in $f$. The weight of
an edge is the difference $f(x)-f(y)$.

The $\ell_{1}$ distance of $f$ to monotonicity $\dist(f, \mathrm{mono})$ is the $\ell_{1}$ distance
of $f$ to the closest real-valued monotone function.

\begin{definition}[Distance to monotonicity]
The $\ell_{1}$ distance of $f: P \to \R$ to monotonicity is its distance to the closest real-valued monotone function. 
\[\dist_1(f, \mathrm{mono}) := \min_{\text{monotone } g: P \to \R} \bigg [ \frac{1}{|P|} \sum_{x \in P} |f(x) - g(x)| \bigg ]\]
The Hamming distance to monotonicity of $f: P \to \bits$ is defined analogously.
\[\dist_0(f, \mathrm{mono}) := \min_{\text{monotone } g: P \to \bits} \bigg [ \frac{1}{|P|} \sum_{x \in P} \Ind[f(x) \ne g(x)] \bigg ]\]
\end{definition}


%
%
%
We will need a bound on how well monotone functions can be approximated by low-degree polynomials. The following fact follows\footnote{see \cite{lange_properly_2023} for more explanation on how these references yield the fact below.} from \cite{bshouty_fourier_1996, kalai_agnostically_2005} and a refinement by \cite{feldman_tight_2020}.
\begin{fact}
\label{fact:A monotone Boolean functions well-approximated by polys}For every monotone $f: \bits^n \to \bits$ and $\eps > 0$, there exists a multilinear polynomial $p$ of degree $\lceil \frac{4 \cdot \sqrt{n}}{\eps} \log \frac{4}{\eps} \rceil$ such that 
\[||f - p||_1 \le \eps.\]
\end{fact}

\subsection{Convex optimization}
\begin{definition}
A \textbf{separation oracle }for a convex set $\mathcal{C}_{\text{convex}}$
is an oracle that given a point $\vect x$ does one of the following
things:
\begin{itemize}
\item If $x\in\mathcal{C}_{\text{convex}}$, then the oracle outputs ``Yes''.
\item If $x\notin\mathcal{C}_{\text{convex}}$, then the oracle outputs $(\text{No}, q_{\text{separation}})$, where $Q_{\text{separation}} \in \R ^d$ represents a direction along which $x$ is separated from $\mathcal{C}_{\text{convex}}$. Formally,  $\langle Q_{\text{separation}}, x\rangle >  \langle Q_{\text{separation}}, x'\rangle$ for any $x'$ in $\mathcal{C}_{\text{convex}}$.
\end{itemize}
\end{definition}
\begin{fact}
\label{fact: ellipsoid algorithm}
There is an algorithm $\textsc{EllipsoidAlgorithm}$ that takes as inputs positive real values $r$ and $R$, and access to a separation oracle for some convex set $\mathcal{C}_{\text{convex}}\subset \{x\in \R^d:\quad \norm{x}\leq R\}$. The algorithm runs in time $\poly \parr{d, \log \frac{R}{r}}$ and either outputs an element in $\mathcal{C}_{\text{convex}}$ or outputs FAIL. Furthermore, if $\mathcal{C}_{\text{convex}}$ contains a ball of radius $r$, the algorithm is guaranteed to succeed.

\end{fact}

\subsection{LCAs and succinct representations}
\label{sec:prelim lca}
We use the following LCAs in this work:
\begin{theorem}[LCA for maximal matching\footnote{To be fully precise, \cite{ghaffari_local_2022} gives an LCA for the task of maximal independent set. The reduction to maximal matching is standard, see e.g. \cite{lange_properly_2022}.} \cite{ghaffari_local_2022}]
\label{cor:matchings} 
There is an algorithm $\text{GhaffariMatching}$ that takes all-neighbor access to a graph $G$, with $N$ vertices and largest degree at most $\Delta$, a random string $r \in \zo^{\poly(\Delta, \log(N/\delta))}$, parameter $\delta \in (0,1)$ and a vertex $v \in G$. The algorithm outputs the identity of a vertex $u:(u,v) \in E(G)$ or $\bot$. The algorithm runs in time $\poly(\Delta, \log(N/\delta))$ and with probability at least $1-\delta$ over the choice of $r$ the condition of \textbf{global consistency holds} i.e. the set of edges $\{(u, v)\in G: \, \text{GhaffariMatching}(G, r, \delta, u)=v\}$ is a maximal matching in the graph $G$.
\end{theorem}

\begin{theorem}[LCA for monotonicity correction of Boolean-valued functions  \cite{lange_properly_2022}]
\label{thm:LRV}
There is an algorithm $\text{BooleanCorrector}$ that takes access to a function $f: P \to \bits$
and all-neighbor access to a poset $P$ with $N$ vertices, 
such that each element has at most $\Delta$ predecessors and successors 
and the longest directed path has length $h$, a random string $r \in \zo^{\poly(\Delta, \log(N/\delta))}$, a parameter $\delta \in (0,1)$ and an element $x$ in $P$. The algorithm outputs a value in $\bits$. The algorithm runs in time $\Delta^{O(\log h)} \cdot \polylog(N/\delta)$ and with probability at least $1-\delta$ over the choice of $r$ the condition of \textbf{global consistency holds} i.e. the function $g: P \to \bits$ defined as 
$
g(x):=\text{BooleanCorrector}(P, r, \delta, x)
$ is monotone and is
such that 
$\Pr_{x \sim P}[g(x) \ne f(x)] \le 2 \cdot \dist(f, \mathrm{mono})$.
\end{theorem}

An important idea in \cite{lange_properly_2022} is that LCAs (i.e. algorithms that achieve global consistency) can be used to operate on succinct representations of combinatorial objects. To explain further, we need the following definition:

\begin{definition}[Succinct representation]
A succinct representation of a function $f$ of size $s$ is a a description of $f$ that is stored in $s$ bits of memory and can be evaluated on an input in $O(s)$ time.  
\end{definition}
For example, circuits of size $s$ and polynomials of degree $\log s$ are examples of succinct representations of size $s$. The following fact follows immediately from the definition:
\begin{fact}[Composition of representations]
\label{fact: rep-composition}
If a function $f$ has a description that uses $t$ bits of memory and evaluates in time $O(t)$ given 
oracle queries to a function $g$, 
and $g$ has a succinct representation of size $s$,
then there is a succinct representation of $f$ of size $O(t + sq)$.
\end{fact}
Now, for example,
combining \footnote{A note on the description sizes of LCAs: because LCAs are uniform (i.e. Turing-machine) algorithms, they can be simulated with a uniform circuit family. For each input size, the size of the corresponding circuit is polynomial in the running time of the LCA for that input size.} \Cref{fact: rep-composition} and \Cref{cor:matchings} we see immediately that for a graph $G$, with $N$ vertices and largest degree at most $\Delta$, using the algorithm in \Cref{cor:matchings} we can transform a size-$s$ representation\footnote{For simplicity, in the rest of the paper we will refer to such function as a "succinct representation of $G$".} of a function computing all-neighbor access to $G$ into a size-$\parr{\Delta^{O(\log h)} \cdot \polylog(N/\delta) \cdot s}$  representation\footnote{For simplicity, in the rest of the paper we will refer to such function simply as "representation of a maximal matching".} of a function that determines membership in some maximal matching over $G$. Note that this transformation itself runs in time $\Delta^{O(\log h)} \cdot \polylog(N/\delta) \cdot s$. Analogously, in an exact same fashion it is possible to combine \Cref{fact: rep-composition} and \Cref{thm:LRV}.




%% file: pseudocode.tex
\section{Our algorithms}
\label{sec:pseudocode}
In this section we give descriptions of the agnostic learning algorithm and its major components (we will analyze the algorithms in the subsequent sections).
The algorithm \textsc{MonotoneLearner} makes calls to \\ \textsc{EllipsoidAlgorithm}, where the optimization domain is the $\leq n^{\left\lceil \frac{4 \cdot \sqrt{n}}{\eps} \log \frac{4}{\eps} \right\rceil}$-dimensional
space of degree-$\left\lceil \frac{4 \cdot \sqrt{n}}{\eps} \log \frac{4}{\eps} \right\rceil$ polynomials over $\R^n$, and its output, $P^\text{GOOD}$, is such a polynomial (see \Cref{fact: ellipsoid algorithm} for details).
It also makes calls to \textsc{HypercubeCorrector}, which is given in \Cref{thm:l1-cube-corrector}. 
\begin{algorithm}[H]
\begin{algorithmic}[1]
\State \textbf{Given:} \parbox[t]{410pt}{Integer $n$, $\epsilon \in (0,1)$, and 
uniform sample access to an unknown function $f:\{\pm 1\}^n \rightarrow
\{\pm 1\}$.\strut}
  
\State \textbf{Output:} Circuit $\mathcal{C}: \{\pm 1\}^n \righ \{\pm 1\}$.
  \For{$\alpha \in \{\epsilon,2\epsilon,3\epsilon,\cdots1-\epsilon,1+200\epsilon\}$}                    
        \State $\text{OptimizationResult} \leftarrow \textsc{EllipsoidAlgorithm}\parr{1, \epsilon \cdot n^{-\frac{1}{2}\left\lceil \frac{4 \cdot \sqrt{n}}{\eps} \log \frac{4}{\eps} \right\rceil},  \textsc{Oracle}
_{\alpha, n, \epsilon}
}$.
        
        \If{OptimizationResult$\neq$FAIL}
        \State $P^\text{GOOD} = \text{OptimizationResult}$
        
\State $P^\text{GOOD}_{\text{TRIMMED}} 
\leftarrow$  representation of a function that takes input 
$\vect x$ and outputs the value
\begin{center}
$\begin{cases}
    P^{\text{GOOD}}(x) &\text{if }P^{\text{GOOD}}(x)\in [-1,+1] \\
    1 &\text{if } P^{\text{GOOD}}(x) > 1\\
    -1 &\text{if } P^{\text{GOOD}}(x) < -1
\end{cases}$
\end{center}
\State $P^\text{GOOD}_{\text{CORRECTED}}\leftarrow$ representation of a function that takes input $\vect x$ and returns the value 
\begin{center}
$\textsc{HypercubeCorrector}(x,P^\text{GOOD}_{\text{TRIMMED}},r)$
\end{center}

\State $T \leftarrow \frac{200}{\epsilon^2} \log\parr{\frac{20}{\epsilon}} \log (20n)$ i.i.d. pairs $(\vect x_i, f(\vect x_i))$, with $\vect x_i$ sampled uniformly from $\bits^n$.  
\State ThresholdCandidates$\leftarrow\left\{\frac{1}{\epsilon} \text{ i.i.d. uniformly random elements in } [-1,1]\right\}$.
\State $t^*:=\argmin_{t \in \text{ThresholdCandidates}} \pars{\frac{1}{|T|}\sum_{\vect x \in T}\left[\abs{\sign(P_{\text{GOOD, TRIMMED, CORRECTED}}(\vect x)-t)-f(\vect x)}\right]}$
\State \Return representation of a function that takes input $\vect x$ and returns the value \label{line:return}
\begin{center}
$\sign(P_{\text{GOOD, TRIMMED, CORRECTED}}(\vect x)-t)^*$ 
\end{center}
        \EndIf 
    \EndFor

\end{algorithmic}
\caption{Algorithm \textsc{MonotoneLearner}
$(n, \epsilon, T)$
}
\label{alg:monotone learner}
\end{algorithm}

The subroutine \textsc{Oracle} takes as input a polynomial and provides the separating hyperplane required by \textsc{EllipsoidAlgorithm}.
It makes calls to \textsc{HypercubeMatching} (see \Cref{lem:cube-matching}), which provides a high-weight matching over the pairs of labels that violate monotonicity.

\begin{algorithm}[H]
\begin{algorithmic}[1]
\State \textbf{Given:} \parbox[t]{410pt}{$\epsilon,\alpha \in (0,1)$,  degree-$\left\lceil \frac{4 \cdot \sqrt{n}}{\eps} \log \frac{4}{\eps} \right\rceil$ polynomial
$P$ over $\R^n$ with $\norm P_{2}\leq1$, and\\
uniform sample access to an unknown function $f:\{\pm 1\}^n \rightarrow
\{\pm 1\}$.
\strut}
  
\State \textbf{Output:} "Yes" or ("No", $Q_{\text{separator}}$), where  $Q_{\text{separator}}$ is a degree-$\left\lceil \frac{4 \cdot \sqrt{n}}{\eps} \log \frac{4}{\eps} \right\rceil$ polynomial over $\R^n$.
\State $P_{\text{TRIMMED}} \leftarrow$  representation of a function that takes input 
$\vect x$ and outputs $\begin{cases}
    P(x) &\text{if }P(x)\in [-1,+1] \\
    1 &\text{if } P(x) > 1\\
    -1 &\text{if } P(x) < -1
\end{cases}$.
\State \label{line: first mention of C} $T\leftarrow$ set of $n^{\frac{C\sqrt{n}}{\epsilon} \log \frac{1}{\epsilon}}$ i.i.d. pairs $(\vect x_i, f(\vect x_i))$, with $\vect x_i$ sampled uniformly from $\bits^n$ (for sufficiently large constant $C$).
\State $r \leftarrow$ string of $2^{C \sqrt{n} (\log n \cdot \log \frac{1}{\epsilon})^C}$ random i.i.d. bits (for sufficiently large constant $C$).
\State \label{line: second mention of C} $M_{\text{separator}}\leftarrow$ representation of a function that takes input $x$ and outputs 
\begin{center}
$\begin{cases}
    0 & \text{if } \textsc{HypercubeMatching}(P_{\text{TRIMMED}}, \eps/4, r) \text{ does not match $x$ to any other vertex} \\
    1 &\text{if } 
    \textsc{HypercubeMatching}(P_{\text{TRIMMED}}, \eps/4, r) \text{ matches $x$ some vertex $z$, s.t. $z \preceq x$}
    \\
     -1 &\text{if } 
    \textsc{HypercubeMatching}(P_{\text{TRIMMED}}, \eps/4, r) \text{ matches $x$ some vertex $z$, s.t. $z \succeq x$}
\end{cases}$
\end{center}


\If{$\frac{1}{|T|}\sum_{\vect x \in T} \pars{M_{\text{separator}}(\vect x)\cdot P_{\text{TRIMMED}}(\vect x)} > 5 \epsilon$} \label{line:estimate}
\State $Q_{\text{separator}} \leftarrow \sum_{S\subset[n]:\:\abs S\leq\left\lceil \frac{4 \cdot \sqrt{n}}{\eps} \log \frac{4}{\eps} \right\rceil}
\parr{
\frac{1}{|T|}\sum_{\vect x \in T} \pars{M_{\text{separator}}(\vect x)\cdot \chi_S(\vect x)} 
}\chi_S$
\State \Return ("No", $Q_{\text{separator}}$)
 \ElsIf{$\frac{1}{|T|}\sum_{\vect x \in T} \pars{\abs{f(\vect x)-P(\vect x)}}>\alpha+50\epsilon$} \label{line:other estimate}
\State $Q_{\text{separator}} \leftarrow 
\sum_{\substack{S\subset[n]
\abs S\leq\left\lceil \frac{4 \cdot \sqrt{n}}{\eps} \log \frac{4}{\eps} \right\rceil
}
}\parr{\E_{\vect x\sim T}\pars{\widehat{P}(S)\chi_{S}(\vect x)\sign(P(\vect x)-f(\vect x))}}\chi_{S}$ \label{line: qseparator}
\State \Return ("No", $Q_{\text{separator}}$)
\Else \State \Return "Yes"
\EndIf

\end{algorithmic}
\caption{Subroutine $\textsc{Oracle}
_{\alpha,n, \epsilon}(P)$
}
\label{alg:monotonicity oracle}
\end{algorithm}

The algorithm \textsc{MatchViolations} finds a high-weight matching on the violation graph of a poset. It is the main component of \textsc{HypercubeMatching}, which is just a wrapper that calls \textsc{MatchViolations} on the truncated cube.
\textsc{FilterEdges} removes vertices that are either incident to $M$ or have weight below the threshold $t$,
and \textsc{GhaffariMatching} is the maximal matching algorithm of \Cref{cor:matchings}. 
More implementation details and analysis are given in \Cref{sec:matching}.

\begin{algorithm}[H]
\begin{algorithmic}

\State \textbf{Given:} Poset $P$ and function $f:~P \to [-1, 1]$ given as succinct representations, weight threshold $\eps$, random seed $r$
\State \textbf{Output:} Succinct representation of a high-weight matching on the violating pairs of $P$ w.r.t. $f$
\If{$\eps < 1/|P|$}
\State $M \gets$ representation of the greedy algorithm that adds each edge $(x,y)$ of $TC(P)$ in decreasing order of $f(x) - f(y)$. 
\Else 
\State $t \gets 2$
\State $M \gets$ representation of a function computing the empty matching
\While{$t > \eps/2$}
\State $P' \gets$ representation of a function that takes input $x$ and outputs\\ \textsc{FilterEdges}$(TC(P), f, t, M, x)$ 
\State $M \gets $ representation of a function that takes input $x$ and outputs $M(x)$ if $M(x) \ne \bot$, otherwise \textsc{GhaffariMatching}$(P', r, x))$
\State $t \gets t/2$
\EndWhile
\EndIf
\State \Return $M$

\end{algorithmic}
\caption{\textsc {MatchViolations}$(P, f, \eps, r)$}
\label{alg:matching}
\end{algorithm}

The following is the core of \textsc{HypercubeCorrector}, given as a ``global overview'' for convenience.
Analysis and local implementation are given in \Cref{sec:corrector}. 
The algorithm corrects monotonicity of a $k$-valued function over a poset.
\textsc{HypercubeCorrector} is a wrapper that discretizes a real-valued function and then calls this corrector with the truncated hypercube as the poset.

\begin{algorithm}[H]
\begin{algorithmic}[1]
		\State \textbf{Given:} Poset $P$ of height $h$, function $f: P \to [k]$
\State \textbf{Output:} monotone function $g: P \to [k]$
\State Let $i \gets 0$
\For{$0 \le i \le \lceil \log k \rceil$}

\State Let $f_i$ be the projection of $f$ onto the $i_{th}$ most significant bit of $k$, i.e. $f_i(x) = 1$ if the $i_{th}$ bit of $f(x)$ is $1$.
\State Let $P_i$ be the poset on the elements of $P$ with the relation 
$$x \prec_{P_i} y :=  x \prec_P y \text{ and } f_j(x) = f_j(y) \text{ for all $j < i$}.$$

\State Let $\pi_i \gets$ {\sc BooleanCorrector}$(f_i, P_i)$.
\State Let $f \gets f \pi_i$.

\EndFor
\State \Return $f$
\end{algorithmic}
\caption{Global view of sorting $k$-valued labels in a poset}
\label{alg:k-valued-global}
\end{algorithm}

%% file: l1-corrector.tex
\section{Analysis of the local corrector}
\label{sec:corrector}
In this section, we prove \Cref{thm:correction-main} by analyzing our algorithm for correcting a real-valued function over a poset in a way that preserves the $\ell_1$ distance to monotonicity within a factor of 2.
This extends the monotonicity corrector of \cite{lange_properly_2022} to handle functions with non-Boolean ranges.

\begin{lemma}[$\ell_1$ correction of $k$-valued functions]
Let $P$ be a poset and $f: P \to [k]$ be $\alpha$-close to monotone in $\ell_1$ distance.
There is an LCA that makes queries to $f$ and outputs queries to $g: P \to [k]$,
such that $g$ is monotone and $||f - g||_1 \le 2\alpha$. 
The LCA makes $(\Delta \log N)^{O(\log h \log k)}$ queries, 
where $\Delta$ is the maximum number of predecessors or successors of any element in $P$, $N$ is the number of vertices, and $h$ is the length of the longest directed path.. 
It uses a random seed of length $\poly(\Delta \log N)$, and succeeds with probability $1 - N^{-10}$. 
\end{lemma}

The following lemmas are used in the proof of correctness of our algorithm. 
Their proofs are deferred to the appendix.

\begin{lemma}[Equivalence of $k$-valued and bitwise monotonicity]
		\label{lem:bitwise}
Let $f: P \to [k]$ be a function and $f_i$ be the projection of $f$ onto the $i_{th}$ most significant bit of $k$, i.e. $f_i(x) = 1$ if the $i_{th}$ bit of $f(x)$ is $1$,
for each $i \in [ \lceil \log k \rceil ]$. 
Let $P_i$ be the poset on the elements of $P$ with the relation 
$$x \prec_{P_i} y :=  x \prec_P y \text{ and } f_j(x) = f_j(y) \text{ for all $j < i$}.$$
Then $f$ is monotone if and only if each $f_i$ is monotone over the corresponding $P_i$.
\end{lemma}

\begin{lemma}[Preservation of closeness to monotone functions]
		\label{lem:l1-error-preservation}
Let $g$ be obtained from $f$ by swapping the labels of a pair $x \prec_P y$ that violates monotonicity. Then for any monotone function $m$, $||g - m||_1 \le ||f - m||_1$. 
\end{lemma}

The corollary follows from repeated application of \Cref{lem:l1-error-preservation} and the triangle inequality.
\begin{corollary}[$\ell_1$ error preservation]
		Let $g$ be obtained from $f$ by a series of swaps of label pairs that violate monotonicity in $f$. Then $||g - f||_1 \le 2 \cdot \dist_1(f, \mathrm{mono})$.
\label{cor:l1-preservation}
\end{corollary}

We also require a modification to the LCA claimed in \Cref{thm:LRV} for correcting Boolean functions.
That algorithm works by performing a sequence of label-swaps on pairs that violate monotonicity in the poset,
then outputting the function value that ends up at the queried vertex $x$.
It can instead track the swaps and output the identity of the vertex that
$x$ receives its final label from.
The modified algorithm can be thought of as an LCA that gives query access to a label permutation.

\begin{fact}[Poset sorting algorithm implicit in \cite{lange_properly_2022}]
Let $P$ be a poset with $N$ vertices such that every element has at most $\Delta$ predecessors and successors, and the longest directed path has length $h$. Let $f:P \to \bits$ be $\alpha$-close to monotone in Hamming distance.
There is an algorithm {\sc BooleanCorrector} that
gives query access to a permutation $\pi$ of $P$ such that $f\pi$ is a monotone function and $\Pr_{x \sim P}[f(x) \ne (f \pi) (x)] \le 2\alpha$. The LCA implementation of {\sc BooleanCorrector} uses $(\Delta \log N)^{O(
\log h)}$ queries and running time, has a random seed of length $\poly(\Delta \log N)$, 
and succeeds with probability $1 - N^{-11}$.
\end{fact}

%
%

Here we present the LCA implementation of \Cref{alg:k-valued-global}.

\begin{algorithm}[H]
\begin{algorithmic}[1]
		\State \textbf{Given:} Target vertex $x$, all-neighbors (immediate predecessor and successor) oracle for $P$, query access to $f: P \to [k]$, iteration number $i$, random seed $r$.
        \State \textbf{Output:} query access to function $g: P\to [k]$ which is monotone when truncated to the first $i$ most significant bits.
        \If{$i=0$} \Return $f(x)$
        \Else
		\State $S \gets $ the set of all predecessors and successors of $x$ in $P$
        \For {$y \in S$}
        \State Let $f'(y) \gets ${\sc $k$-Corrector}$(y, P, f, i-1, r)$.
        \EndFor
        \State Let $f'_i$ be defined as in \Cref{alg:k-valued-global}, and $P'_i$ be similarly defined with respect to $f'_i$.
		\State Remove any $y$ from $S$ such that $f'_i(y) = f'_i(x)$ or $y$ and $x$ are incomparable in $P'_i$.
        \State Let $z \gets$ {\sc BooleanCorrector}$(x, P'_i, f'_i, r)$
        \State \Return $f'(z)$
    \EndIf

\end{algorithmic}
\caption{LCA implementation of \Cref{alg:k-valued-global}, {\sc $k$-Corrector}$(x, P, f, i, r)$ }
\label{alg:local-k-corrector}
\end{algorithm}

\begin{algorithm}[H]
\begin{algorithmic}
		\State \textbf{Given:} function $f: \bits \to [-1,1]$ given as succinct representation, additive error parameter $\eps>0$, random seed $r$.
        \State \textbf{Output:} succinct representation of monotone function $g: \bits \to [-1,1]$.
        \State $P \gets $ representation of a function that takes $x$ and outputs {\sc TruncatedCube}$(x, \eps)$
        \State $f' \gets $ representation of a function that takes $x$ and outputs $\lfloor f(x) / \eps \rfloor $
        
        \State $f'' \gets $ representation of a function that takes $x$ and outputs \\$
        \indent \begin{cases}
        \eps \cdot k\text{\sc-Corrector}(x, P, f', \lceil \log(1/\eps) \rceil, r) &- \sqrt{2n \log 2/\eps} \le |x| \le\sqrt{2n \log 2/\eps} \\
        1 & |x| \ge \sqrt{2n \log 2/\eps} \\
        -1 & |x| \le - \sqrt{2n \log 2/\eps} 
        \end{cases}$
        
        \State \Return $f''$

\end{algorithmic}
\caption{{\sc HypercubeCorrector}$(f, \eps, r)$ }
\label{alg:cube-corrector}
\end{algorithm}

\begin{lemma}[Correctness and query complexity of \Cref{alg:local-k-corrector}]
\label{lem:k-correctness}
With probability $1 - i \cdot N^{-11}$ over a random seed $r$ of length $\poly(\Delta \log N)$,
the algorithm {\sc $k$-Corrector}$(x, P, f, i, r)$ gives query access to a function $g$ that is monotone when truncated to the first $i$ most significant bits. 
Its query complexity is $(\Delta \log N)^{O(i \log h + 1)}$, and $||g - f||_1 \le 2\alpha$, where $\alpha$ is the $\ell_1$ distance of $f$ to the nearest monotone function.
\end{lemma}

\begin{proof}
Fix the random seed $r$ and assume all calls to {\sc BooleanCorrector} succeed with $r$,
then we proceed by induction. 
In the base case, $f$ is certainly monotone when truncated to 0 bits and the algorithm makes only 1 query. 
In the inductive case, suppose the claim holds for $i-1$; in other words $k$-{\sc Corrector}$(y, P, f, i-1, r)$ makes $(\Delta \log N)^{O((i-1) \log h + 1)}$ queries
and returns a function that is monotone in the first $i-1$ bits. 
Then when $k$-{\sc Corrector} is called with iteration number $i$, the function $f'_j$ is monotone over $P'_j$ for all $j < i$.
{\sc BooleanCorrector}$(x, P'_i, f'_i, r)$ returns a vertex to swap labels with $x$ such that the resulting function is monotone in the $i_{th}$ bit, over the poset $P'_i$.  
Then the function returned by $k$-{\sc Corrector} satisfies the conditions of \Cref{lem:bitwise} for the first $i$ bits, so it must be monotone in the first $i$ bits. 

We now bound the failure probability and distance to $f$.
The failure probability of {\sc BooleanCorrector} is $N^{-11}$
and we call {\sc BooleanCorrector} on $i$ different graphs, so by union bound the total failure probability is $\le i \cdot N^{-11}$ as desired.
The fact that $||g - f||_1 \le 2\alpha$ follows from \Cref{cor:l1-preservation}.
\end{proof}

We can now prove \Cref{thm:correction-main}.

\correctorthm*

\begin{proof}[Proof of \Cref{thm:correction-main}]
Given some $\eps \in (0,1/2)$, let $f_\eps(x) := \lfloor f(x)/\eps \rfloor$; 
certainly queries to $f_\eps$ can be simulated by queries to $f$. 
On input $x$, run $k$-\textsc{Corrector}$(x, P, f_\eps, \lceil \log(2/\eps) \rceil, r)$ with a random seed $r$ of length $\poly(\Delta \log N)$.
By \Cref{lem:k-correctness}, this makes $(\Delta \log N)^{O(\log(1/\eps) \log h )}$ queries to $f_\eps$ and outputs $g_\eps(x)$, where $g$ is monotone and $||g_\eps - f_\eps||_1 \le 2 \cdot \dist_1(f_\eps, \mathrm{mono})$.
Since $f$ is $\alpha$-close to some monotone function $m$, 
we have $\dist_1(f_\eps, \mathrm{mono}) \le ||f_\eps - m/\eps||_1 \le ||f/\eps - m/\eps||_1 + ||f/\eps - f_\eps||_1 \le \alpha/\eps +1$.\\

Return $g(x) := \eps \cdot g_\eps(x)$. 
Then $$||g - f||_1 = ||\eps g_\eps - f||_1 \le ||\eps g_\eps - \eps f_\eps ||_1 + ||\eps f_\eps - f||_1 \le 2\eps(\alpha/\eps + 1) + \eps \le 2\alpha + 3\eps.$$

The failure probability is $N^{-11} \cdot \lceil \log(2/\eps) \rceil$ by \Cref{lem:k-correctness}, 
but we will assume that $\lceil \log (2/\eps) \rceil < N$.
Otherwise, the allowed query complexity and running time would exceed $\Delta^N$, which is $> \Delta N$ for any $\Delta, N > 1$.
With $O(\Delta N)$ query complexity and running time, a trivial algorithm would suffice:
one could solve the linear program with $\Delta N$ monotonicity constraints,
minimizing $||g - f||_1$. 
Under our assumption, the failure probability is at most $N^{-10}$. 
\end{proof}

\begin{corollary}[Monotonizing a representation of a function on the Boolean cube]
\label{thm:l1-cube-corrector}
Let $f: \bits^n \to [-1,1]$ be $\alpha$-close to monotone in $\ell_1$ distance, given as a succinct representation of size $s_f$. 
There is an algorithm that runs in time $2^{\tilde{O}(\sqrt{n} \log^{3/2}(1/\eps))} \cdot s_f$ time and outputs a monotone function $g$ such that $||f - g||_1 \le 2\alpha + 4\eps$. The size of the representation of $g$ is $2^{\tilde{O}(\sqrt{n}\log^{3/2}(1/\eps))} \cdot s_f$. The algorithm uses a random seed of length $2^{\tilde{O}(\sqrt{n} \log (1/\eps))}$ and succeeds with probability $1 - 2^{-10n}$.
\end{corollary}
The proof of \Cref{thm:l1-cube-corrector} is deferred to \Cref{apx:corrector}.

\section{Analysis of the matching algorithm}
\label{sec:matching}

In this section we give an algorithm for generating a succinct representation of a matching over the violated pairs of the hypercube whose weight is a constant factor of the distance to monotonicity. The core of the algorithm is an LCA for finding such a matching over the violated pairs of an arbitrary poset.

\begin{lemma}[Equivalence of distance to monotonicity and maximum-weight matching]
		\label{lem:matching-approx-dist}
Let $W$ be the total weight of the maximum-weight matching of the violation graph of $f$. 
Then $\dist_1(f, \mathrm{mono}) = W / N$. 
\end{lemma}
\begin{proof}
This proof is analogous to the proof of Lemma 3.1 of \cite{berman_l_p-testing_2014}; see \Cref{apx:matching}.
\end{proof}

\subsection{Details and correctness of \textsc{MatchViolations}}
The algorithm \textsc{Matchviolations} given in \Cref{sec:pseudocode} makes calls to an algorithm called \textsc{FilterEdges}, which removes vertices that have already been matched or are not incident to any heavy edges.
We give the pseudocode for \textsc{FilterEdges} here.

\begin{algorithm}[H]
\begin{algorithmic}[1]
\label{alg:filter edges}

\State \textbf{Given:} Poset $P$, function $f:~P \to [-1, 1]$, and matching $M$ given as succinct representations, weight threshold $t$, vertex $x$
\State \textbf{Output:} All neighbors of $x$ in the graph of violation score $\ge t$ and not in $M$
\State \Return \begin{multline*}\{y \in P(x)~|~ \\
M(y) = \bot \text{ and } \left [(x < y\text{ and }f(x) \ge f(y) + t)\text{ or }(x > y\text{ and }f(x) \le f(y) - t) \right ]\}
\end{multline*}
%

\end{algorithmic}
\caption{LCA: \textsc{FilterEdges}$(P, f, t, M, x)$}
\end{algorithm}

%
%

\begin{lemma}
\label{lem:global matching}
Let $P$ be a poset with $N$ vertices, and let $\Delta$ be an upper bound on the number of predecessors and successors of any vertex in $P$.
Then the output of the LCA \textsc{MatchViolations}$(P, f, \eps, r)$ with a random seed $r$ of length $\poly(\Delta, \log N)$,
is a matching of weight at least $N(\lfrac{1}{4}\dist_1(f, \mathrm{mono}) - \eps)$ with probability at least $1 - N^{-10}$. 
  
\end{lemma}

\begin{proof}
This is a small modification to the standard greedy algorithm for high-weight matching; see \Cref{apx:matching}.
\end{proof}

%


\begin{lemma}[Running time and output size]
\label{lem:matching complexity}
Let $P, f, \eps, N, \Delta$, and $r$ be as described in the lemma above.
Let $s_P$ be the size of the succinct representation of $P$, and $s_f$ be the size of the succinct representation of $f$. 

Then \textsc{MatchViolations}$(P, f, \eps, r)$ runs in time $(\Delta \log N)^{O(\log(1/\eps))} (s_P + s_f)$ and outputs a representation of size $(\Delta \log N)^{O(\log(1/\eps))} (s_P + s_f)$.
\end{lemma}
\begin{proof}
If $\eps < 1/N$, then \textsc{MatchViolations} constructs and outputs a representation of the standard global greedy algorithm for 2-approximate maximum matching. 
The representation size of this algorithm is $O(\Delta N) \le (\Delta \log N)^{O(\log(1/\eps))}$, and the running time of \textsc{MatchViolations} is polynomial in this representation size. 

If $\eps \ge 1/N$, then by induction on the number of iterations $i$, we will show that the representation size of $M$ at the start of iteration $i$ is at most $(\Delta \log N)^{O(i)} (s_P + s_f)$.
In the base case, we have an empty matching $M$ which has constant representation size.

In the inductive case, suppose the claim holds at the start of iteration $i$. 
Then we set $P'$ to be the function that applies \textsc{FilterEdges} to $TC(P)$.
$TC(P)$ has size $O(\Delta \cdot s_P)$, as it makes $O(\Delta)$ calls to $P$.
\textsc{FilterEdges} makes one call to $TC(P)$ and 
at most $O(\Delta)$ calls to $M$ and $f$. 
It also has overhead of size $O(\log t) = O(\log(1/\eps)) = O(\log N)$.
By the inductive hypothesis, the size of $P'$ is then 
\begin{align*}
    &O(\Delta) \cdot (\Delta \log N)^{O(i)} (s_P + s_f ) + O(\log N) + O(\Delta \cdot s_P) \\
    \le & (\Delta \log N)^{O(i+1)} (s_P + s_f).
\end{align*}
Then we set $M$ to be the function that applies \textsc{GhaffariMatching} to $P'$.
\textsc{GhaffariMatching} has constant overhead and makes $\poly(\Delta, \log N)$ queries to $P'$.
Then the new size of $M$ is 
$\poly(\Delta, \log N) \cdot (\Delta \log N)^{O(i)} (s_P + s_f)= (\Delta \log N)^{O(i+1)} (s_P + s_f)$.

The size bounds follow from the fact that there are $O(\log 1/\eps)$ iterations. 
The corresponding running time bound for \textsc{MatchViolations} comes from the fact that since it only constructs the succinct representations, its running time in each iteration is polynomial in the size of the representations it constructs.

\end{proof}

\begin{algorithm}[H]
\begin{algorithmic}[1]

\State \textbf{Given:} Function $f:~\bits^n \to [-1, 1]$ given as succinct representation, weight threshold $\eps$, random seed $r$
\State \textbf{Output:} Succinct representation of a high-weight matching on the violating pairs w.r.t. $f$
\State $P \gets \textsc{TruncatedCube}(n, \eps)$
\State $M \gets$ representation of a function that takes $x$ and outputs  \\
$\begin{cases}
		\textsc{MatchViolations}(P, f, \eps, r) &  - \sqrt{2n \log 2/\eps} \le |x| \le  \sqrt{2n \log 2/\eps} \\
		\bot & \text{otherwise}
\end{cases}$
\State \Return $M$
\end{algorithmic}
\caption{\textsc {HypercubeMatching}$( f, \eps, r)$}
\label{alg:cube-matching}
\end{algorithm}

\begin{lemma}
\label{lem:cube-matching}
With a random seed of length $2^{\tilde{O}(\sqrt{n} \log(1/\eps))}$, \Cref{alg:cube-matching} outputs a representation of a matching on the weighted violation graph $\viol(f)$, of weight at least $2^n \cdot (\frac{1}{4}\dist_1(f, \mathrm{mono}) - 4\eps)$, with probability at least $1 - 2^{-10n}$.
The size of the representation is $2^{\tilde{O}(\sqrt{n} \log(1/\eps))} \cdot s_f$, where $s_f$ is the size of the representation of $f$. 
\end{lemma}
\begin{proof}
\textsc{HypercubeMatching} calls \textsc{MatchViolations} on the truncated hypercube,
which has parameters $N < 2^n$ and $\Delta = 2^{O(\sqrt{n} \log n \log (1/\eps))}$. 
The size of the representation of \textsc{TruncatedCube} is $O(n)$.
So by \Cref{lem:matching complexity}, the running time and output size of \textsc{HypercubeMatching} are  $2^{O(\sqrt{n} \log n \log(1/\eps))} \cdot s_f$,
and the random seed length is $2^{O(\sqrt{n} \log n \log(1/\eps))}$. 

Let $f'$ be the restriction of $f$ to the truncated cube.
Since $f$ is bounded in $[-1,1]$ and the truncated cube covers all but an $\eps$ fraction of vertices, we have $\dist_1(f', \mathrm{mono}) \ge \dist_1(f, \mathrm{mono}) - 2\eps$.
By \Cref{lem:global matching}, the weight of the matching is at least $(1-\eps) \cdot 2^n(\lfrac{1}{4}\dist_1(f', \mathrm{mono}) - \eps) \ge (1-\eps) \cdot 2^n(\lfrac{1}{4} \dist_1(f, \mathrm{mono}) - 3\eps/2) \ge 2^n( \lfrac{1}{4} \dist_1(f, \mathrm{mono}) - 4\eps)$.

\end{proof}

%% file: TheEllipsoidThing.tex
\section{Analysis of the agnostic learning algorithm}

By inspecting algorithm \textsc{MonotoneLearner} (i.e. \Cref{alg:monotone learner} on page \pageref{alg:monotone learner}), we see immediately that the run-time is $2^{\tilde{O}(\sqrt{n}/\epsilon)}$. We proceed to argue that the algorithm indeed satisfies the guarantee of \Cref{theorem: main theorem}.
First, we will need the following standard proposition. 
\begin{restatable}{claim}{concentration}
\label{claim: big enough set of samples preserves L_1 distances}
For any positive integers $n$ and $d$, real $\epsilon, \delta\in (0,1)$, and any function
$f:\left\{ \pm1\right\} ^{n}\righ[-1,1]$, let
$T$ be a collection of at least $n^{5d} \cdot \frac{100}{\epsilon^{2}}\ln\frac{1}{\epsilon}\ln\frac{1}{\delta}$
i.i.d. uniformly random elements of $\left\{ \pm1\right\} ^{n}$.
Then, with probability at least $1-\delta$ 
\[
\max_{\substack{\text{degree-\ensuremath{d} polynomial \ensuremath{P} over \ensuremath{\left\{  \pm1\right\}  ^{n}}}\\
\text{with \ensuremath{\norm P_{2}\leq}1}
}
}\abs{\norm{f-P}_{1}-\E_{\vect x\sim T}\pars{\abs{f(\vect x)-P(\vect x)}}}\leq\epsilon,
\]
\end{restatable}
\begin{proof}
    See \Cref{subsec: proof of claim big enough set of samples preserves L_1 distances} for the proof of this proposition.
\end{proof}

Now, in the following lemma we prove that subroutine $\textbf{Oracle}
_{\alpha,n, \epsilon}(P)$ (i.e. \Cref{alg:monotonicity oracle} on page \pageref{alg:monotonicity oracle}) satisfies some precise specifications  with high probability. Informally, we show that $\textbf{Oracle}
_{\alpha,n, \epsilon}(P)$ either 
\begin{itemize}
    \item Certifies that the polynomial $P$ is both
close to monotone in $L_1$ distance and has $L_1$ prediction error of $\alpha + O(\epsilon)$.
    \item Outputs a hyperplane separating $P$ from all such polynomials.
\end{itemize}
Formally, we prove the following:
\begin{lemma}
\label{lem: main lemma-1} For sufficiently large constant $C$ in \Cref{line: first mention of C} and \Cref{line: second mention of C} of procedure
$\textbf{Oracle}
_{\alpha,n, \epsilon}(P)$, sufficiently large integer $n$, any function $f:\bits^n \righ \bits$, parameters $\epsilon, \alpha\in (0,1)$, and a degree-$\left\lceil \frac{4 \cdot \sqrt{n}}{\eps} \log \frac{4}{\eps} \right\rceil$
polynomial $P$ satisfying $\norm P_{2}\leq1$ the following is true. 
The procedure
$\textbf{Oracle}
_{\alpha,n, \epsilon}(P)$
runs in time $n^{\tilde{O}{\parr{\frac{\sqrt{n}}{\epsilon}}}}$ and
will with probability at least $1-\frac{1}{2^{5n}}$ conform to the
following specification: 
\begin{enumerate}
\item If $\textbf{Oracle}
_{\alpha,n, \epsilon}(P)$ outputs ``yes'', then:
\begin{enumerate}
\item The function 
$P_{\text{TRIMMED}}=\begin{cases}
1 & \text{if \ensuremath{P(x)>1,}}\\
-1 & \text{if }\ensuremath{P(x)<-1,}\\
P(\vect x) & \text{otherwise}.
\end{cases}$

is $100\epsilon$-close to monotone in $L_{1}$ norm.
\item The $L_{1}$ distance between $P$ and the function $f$ is at most $\alpha+100\epsilon$.
\end{enumerate}
\item If $\textbf{Oracle}
_{\alpha,n, \epsilon}(P)$ instead outputs ("No", $Q_{\text{separator}}$), where  $Q_{\text{separator}}$ is a degree-$\left\lceil \frac{4 \cdot \sqrt{n}}{\eps} \log \frac{4}{\eps} \right\rceil$ polynomial over $\R^n$,
then we have $\left\langle P',Q_{\text{separator}}\right\rangle <\left\langle P,Q_{\text{separator}}\right\rangle $
for any degree-$\left\lceil \frac{4 \cdot \sqrt{n}}{\eps} \log \frac{4}{\eps} \right\rceil$ polynomial $P'$
with $\norm{P'}_{2}\leq1$ that satisfies the following two conditions:
\begin{itemize}
\item $P'$ is $\epsilon$-close in $L_{1}$ distance to some monotone function
$f_{\text{monotone}}:\left\{ \pm1\right\} ^{n}\righ\left[-1,1\right]$
and
\item $P'$ is $\parr{\alpha+\epsilon}$-close in $L_{1}$ distance to the
function $f$ which we are trying to learn.
\end{itemize}
\end{enumerate}
In particular, this implies that if polynomial $P$ itself is $\epsilon$-close
in $L_{1}$ distance to some monotone function and is $\parr{\alpha+\epsilon}$-close
in $L_{1}$ distance to the function $f$, then $\textbf{Oracle}
_{\alpha,n, \epsilon}(P)$
will say ``yes'' with probability at least $1-\frac{1}{2^{10n}}$.
\end{lemma}

\begin{proof}
We use the union bound to conclude that with probability at least $1-\frac{1}{2^{5n}}$
all the following events hold:
\begin{itemize}
\item The LCA from \Cref{lem:cube-matching} works as advertised
and the weight $W$ of the resulting matching satisfies 
\[
\frac{W}{2^{n}}\geq0.1 \,
\dist_1(P_{\text{TRIMMED}}, \mathrm{mono})
-\epsilon.
\]
Another way to write the same thing is
\begin{equation}
\left\langle M_{\text{separator}},P_{\text{TRIMMED}}\right\rangle \geq0.1
 \,
\dist_1(P_{\text{TRIMMED}}, \mathrm{mono})
-\epsilon.\label{eq: matching approximately largest weight}
\end{equation}
From
\Cref{lem:cube-matching} 
it follows that this holds with probability at least $1-\frac{1}{2^{10n}}$.
\item The estimate of $\left\langle M_{\text{separator}},P_{\text{TRIMMED}}\right\rangle $
in \Cref{line:estimate} is indeed $\epsilon$-close to the true value. From the standard Hoeffding bound, this holds with probability at least $1-\frac{1}{2^{10n}}$.
\item It is the case that
\[
\norm{
\sum_{S\subset[n]:\:\abs S\leq\left\lceil \frac{4 \cdot \sqrt{n}}{\eps} \log \frac{4}{\eps} \right\rceil}\widehat{M_{\text{separator}}}(S)\chi_S
-
Q_{\text{separator}}
}_2
\leq
\epsilon
\]
Substituting the expression for $Q_{\text{separator}}$, and using the orthogonality of $\{\chi_S\}$ we see this is equivalent to
\[
\sum_{S\subset[n]:\:\abs S\leq\left\lceil \frac{4 \cdot \sqrt{n}}{\eps} \log \frac{4}{\eps} \right\rceil}
\underbrace{
\parr{\widehat{M_{\text{separator}}}(S)
-
\frac{1}{|T|}\sum_{\vect x \in T} \pars{M_{\text{separator}}(\vect x)\cdot \chi_S(\vect x)} 
}^2}_{\text{$\leq \epsilon n^{-\left\lceil \frac{4 \cdot \sqrt{n}}{\eps} \log \frac{4}{\eps} \right\rceil}$ in absolute value w.p. $\geq\frac{1}{2^{10n}}$ via Hoeffding's bound}
}
\leq
\epsilon
\]
Overall, the above holds with probability at least $1-\frac{1}{2^{9n}}$ by taking a  Hoeffding bound for each individual summand and taking a union bound over them.
\item The set $T\subset\left\{ \pm1\right\} ^{n}$ is
such that 
\begin{equation}
\max_{\substack{\text{degree-\ensuremath{\left\lceil \frac{4 \cdot \sqrt{n}}{\eps} \log \frac{4}{\eps} \right\rceil} polynomial \ensuremath{P'} over \ensuremath{\left\{  \pm1\right\}  ^{n}}}\\
\text{with \ensuremath{\norm{P'}_{2}\leq}1}
}
}\abs{\norm{f-P'}_{1}-\E_{\vect (x, f(x))\sim T}\pars{\abs{f(\vect x)-P'(\vect x)}}}\leq\epsilon.\label{eq: empirical L_1 error is close to actual L_1 error}
\end{equation}

It follows from \Cref{claim: big enough set of samples preserves L_1 distances}
that this happens with probability at least to $1-\frac{1}{2^{10n}}$.

\end{itemize}
Now, we argue that if these conditions indeed hold, then $\textbf{Oracle} _{\alpha,n, \epsilon}(P)$
will satisfy the specification given.

First, suppose $\textbf{Oracle} _{\alpha,n, \epsilon}(P)$ answered ``yes''. Then,
since the estimate of $\left\langle M_{\text{separator}},P_{\text{TRIMMED}}\right\rangle $
in \Cref{line:estimate} is within $\epsilon$ of its true value, we have 
\[
\left\langle M_{\text{separator}},P_{\text{TRIMMED}}\right\rangle \leq6\epsilon.
\]
Now, since we are assuming the matching LCA from \Cref{lem:cube-matching} works as advertised, this means that 
\[
6\epsilon\geq\left\langle M_{\text{separator}},P_{\text{TRIMMED}}\right\rangle \geq0.1\cdot \dist_1(P_{\text{TRIMMED}}, \mathrm{mono})-\epsilon
\]
which can be rewritten as 
\[
\dist_1(P_{\text{TRIMMED}}, \mathrm{mono}) \leq70\epsilon\leq100\epsilon,
\]
which is one of the two things we wanted to show. The other one was
showing that the $L_{1}$ distance between $P$ and the function $f$,
which we are trying to learn, is at most $\alpha+100\epsilon$. Since
the algorithm returned ``yes'', it has to be that in \Cref{line:other estimate}
we have 
\[
\E_{\vect x\sim T}\pars{\abs{f(\vect x)-P(\vect x)}}\leq\alpha+50\epsilon.
\]
From \Cref{eq: empirical L_1 error is close to actual L_1 error}
it then follows that 
\[
\norm{f-P}_{1}\leq\E_{\vect x\sim T}\pars{\abs{f(\vect x)-P(\vect x)}}+\epsilon\leq\alpha+51\epsilon\leq\alpha+100\epsilon,
\]
which is the other condition we wanted to show for the case when the
oracle says ``yes''.

Now, assume the oracle outputs ``no'' along with some polynomial $Q_{\text{separator}}$
and let $P'$ be a degree\-$\left\lceil \frac{4 \cdot \sqrt{n}}{\eps} \log \frac{4}{\eps} \right\rceil$ polynomial
with $\norm{P'}_{2}\leq1$ that satisfies the following two conditions\footnote{If no polynomial satisfying these conditions exists, the statement
we are seeking to prove holds vacuously.}:
\begin{itemize}
\item $P'$ is $\epsilon$-close in $L_{1}$ distance to some monotone function
$f_{\text{monotone}}:\left\{ \pm1\right\} ^{n}\righ\left[-1,1\right]$
and
\item $P'$ is $\parr{\alpha+\epsilon}$-close in $L_{1}$ distance to the
function $f$ which we are trying to learn.
\end{itemize}
Here, again, there are two cases. First, suppose we have the case where $Q_{\text{separator}}$
is generated from $M_{\text{separator}}$. We have that the oracle's
estimate of $\left\langle M_{\text{separator}},P_{\text{TRIMMED}}\right\rangle $
is at least $5\epsilon$, which means that $\left\langle M_{\text{separator}},P_{\text{TRIMMED}}\right\rangle \geq4\epsilon$.
We know that $P'$ is $\epsilon$-close in $L_{1}$ distance to some
monotone function $f_{\text{monotone}}:\left\{ \pm1\right\} ^{n}\righ\left[-1,1\right]$.
Since $M_{\text{separator}}$ is defined to be so for every matched
pair $(\vect x_{i},\vect y_{i})$ with $\vect x_{i}\prec\vect y_{i}$
we have $M_{\text{\text{separator}}}(\vect x_{i})=1$ and $M_{\text{\text{separator}}}(\vect y_{i})=-1$
and is $0$ otherwise, and for each such pair $f_{\text{monotone}}\left(\vect{x_{i}}\right)\leq f_{\text{monotone}}\left(\vect{y_{i}}\right)$
we have $\left\langle M_{\text{separator}},f_{\text{monotone}}\right\rangle \leq0$.
This allows us to conclude 
\begin{multline*}
0\geq\left\langle M_{\text{separator}},f_{\text{monotone}}\right\rangle =\left\langle M_{\text{separator}},P'\right\rangle +\left\langle M_{\text{separator}},f_{\text{monotone}}-P'\right\rangle \geq\\
\left\langle M_{\text{separator}},P'\right\rangle -\parr{\max_{x \in \bits^n }\abs{M_{\text{separator}}(x)}}\norm{f_{\text{monotone}}-P'}_{1}\geq\left\langle M_{\text{separator}},P'\right\rangle -\epsilon,
\end{multline*}
which means
\begin{multline}
\epsilon\geq\left\langle M_{\text{separator}},P'\right\rangle =\left\langle \sum_{S\subset[n]:\:\abs S\leq\left\lceil \frac{4 \cdot \sqrt{n}}{\eps} \log \frac{4}{\eps} \right\rceil}\widehat{M_{\text{separator}}}(S)\parr{\prod_{i\in S}x_{i}},P'\right\rangle =\\
\left\langle Q_{\text{separator}},P'\right\rangle -\norm{Q-\sum_{S\subset[n]:\:\abs S\leq\left\lceil \frac{4 \cdot \sqrt{n}}{\eps} \log \frac{4}{\eps} \right\rceil}\widehat{M_{\text{separator}}}(S)\parr{\prod_{i\in S}x_{i}}}_{2}\norm{P'}_{2}\geq\left\langle Q_{\text{separator}},P'\right\rangle -\epsilon.\label{eq: upper bound on <Q, good polynomial>}
\end{multline}
On the other hand, the oracle's estimate of $\left\langle M_{\text{separator}},P_{\text{TRIMMED}}\right\rangle $
is at least $5\epsilon$, which means that it is the case that $\left\langle M_{\text{separator}},P_{\text{TRIMMED}}\right\rangle \geq4\epsilon$.
This allows us to conclude 
\begin{multline}
4\epsilon\leq\overbrace{\left\langle M_{\text{separator}},P_{\text{TRIMMED}}\right\rangle \leq\left\langle M_{\text{separator}},P\right\rangle }^{\substack{\text{Trimming the values of a function }\\
\text{only decreases weights of violated edges.}
}
}=\left\langle \sum_{S\subset[n]:\:\abs S\leq\left\lceil \frac{4 \cdot \sqrt{n}}{\eps} \log \frac{4}{\eps} \right\rceil}\widehat{M_{\text{separator}}}(S)\parr{\prod_{i\in S}x_{i}},P\right\rangle \leq\\
\left\langle Q_{\text{separator}},P\right\rangle +\norm{Q-\sum_{S\subset[n]:\:\abs S\leq\left\lceil \frac{4 \cdot \sqrt{n}}{\eps} \log \frac{4}{\eps} \right\rceil}\widehat{M_{\text{separator}}}(S)\parr{\prod_{i\in S}x_{i}}}_{2}\norm P_{2}\geq\left\langle Q_{\text{separator}},P\right\rangle +\epsilon.\label{eq: lower bound bound on <Q, P>}
\end{multline}
Combining Equation \ref{eq: lower bound bound on <Q, P>} and Equation \ref{eq: upper bound on <Q, good polynomial>}
we get 
\[
\left\langle Q_{\text{separator}},P'\right\rangle \leq2\epsilon<3\epsilon\leq\left\langle Q_{\text{separator}},P\right\rangle 
\]
as required. 

Finally, we consider the case when $Q_{\text{separator}}$ is generated on \Cref{line: qseparator}. Since $P'$ is $\parr{\alpha+\epsilon}$-close
in $L_{1}$ distance to the function $f$, by \Cref{eq: empirical L_1 error is close to actual L_1 error}
we have that 
\[
\alpha+\epsilon\leq\norm{f(\vect x)-P'(\vect x)}_{1}\leq\E_{(\vect  x, f(\vect x))\sim T}\pars{\abs{f(\vect x)-P'(\vect x)}}-\epsilon,
\]
which we can rewrite as $\E_{ (\vect  x, f( \vect x))\sim T}\pars{\abs{f(\vect x)-P'(\vect x)}}\leq\alpha+2\epsilon$.
At the same time, we have \\ $\E_{(\vect  x, f(\vect  x))\sim T}\pars{\abs{f(\vect x)-P(\vect x)}}>\alpha+50\epsilon$,
which means that 
$$\E_{(\vect x, f(\vect x))\sim T}\pars{\abs{f(\vect x)-P(\vect x)}}>\E_{ (\vect x, f(\vect x))\sim T}\pars{\abs{f(\vect x)-P'(\vect x)}}.$$
Therefore, as the function mapping a polynomial $H$ to the value
$\E_{(\vect  x, f(\vect  x))\sim T}\pars{\abs{f(\vect x)-H(\vect x)}}$ is convex
,
it has to be the case that\footnote{To be fully precise, the expression above is a subgradient of the convex function mapping a polynomial $H$ to $\E_{(\vect  x, f(\vect x))\sim T}\pars{\abs{f(\vect x)-H(\vect x)}}$.} 
\begin{multline*}
\left\langle P'-P,
\sum_{\substack{S\subset[n]~:~
\abs S\leq\left\lceil \frac{4 \cdot \sqrt{n}}{\eps} \log \frac{4}{\eps} \right\rceil
}
}\parr{\E_{\vect x\sim T}\pars{\widehat{P}(S)\chi_{S}(\vect x)\sign(P(\vect x)-f(\vect x))}}\chi_{S}
\right\rangle =
\\
=
\left\langle P'-P,\nabla_{H}\parr{\E_{(\vect x, f(\vect x))\sim T}\pars{\abs{H(\vect x)-f(\vect x)}}}\bigg\vert_{H=P}\right\rangle <0.
\end{multline*}
This implies that $\langle Q_{\text{separator}}, P' \rangle \le \langle Q_{\text{sepatator}}, P \rangle$, which completes the proof.
\end{proof}

\subsection{Finishing the proof of the Main Theorem (\Cref{theorem: main theorem}).}

Recall that earlier by inspecting \Cref{alg:monotone learner} we concluded that this algorithm runs in time $2^{\tilde{O}\parr{\frac{\sqrt{n}}{\epsilon}}}$.
Here we use \Cref{lem: main lemma-1} to finish the proof of \Cref{theorem: main theorem} by showing that 
 with probability at least $1-\frac{1}{2^{n}}$ the function  $\sign(P^\text{GOOD}_{\text{TRIMMED}}(\vect x)-t^{*})$
is monotone and is $\text{\text{opt}+}O(\epsilon)$-close to $f$ (where
$\text{opt}$ is the distance of $f$ to the closest monotone function). 

We can further conclude that
with probability at least $1-\frac{1}{2^{3n}}$ the following events hold:
\begin{enumerate}
\item Every time an oracle $\textbf{Oracle}
_{\alpha,n, \epsilon}$ is invoked (for various values
of $\alpha$), its behavior will conform to the specifications in
\Cref{lem: main lemma-1}.

\item The algorithm HypercubeCorrector from \Cref{thm:l1-cube-corrector}
used on line 11 works as advertised, so the function $P^\text{GOOD}_{\text{CORRECTED}}:\left\{ \pm1\right\} \righ\left[-1,1\right]$
is monotone and we indeed have 
\begin{equation}
\norm{P^\text{GOOD}_{\text{CORRECTED}}-P^\text{GOOD}_{\text{TRIMMED}}}_1 \leq 
10 \cdot \dist_1(P^\text{GOOD}_{\text{TRIMMED}}, \mathrm{mono}) +\epsilon.\label{eq: Jane's corrector works}
\end{equation}
\item In step (4), the function $\sign(P^\text{GOOD}_{\text{CORRECTED}}(\vect x)-t^{*})$
satisfies the guarantee from \Cref{fact: rounding}, i.e. 
\begin{equation}
\Pr_{\vect x\sim\left\{ \pm1\right\} ^{n}}\left[\sign(P^\text{GOOD}_{\text{CORRECTED}}(\vect x)-t^{*})\neq f\right]\leq
\frac{1}{2}
\norm{P^\text{GOOD}_{\text{CORRECTED}}-f}_1
+\epsilon\label{eq: rounding works right}
\end{equation}
\end{enumerate}
We argue that each of these events takes place with probability at least $1-\frac{1}{2^{4n}}$:
\begin{itemize}
    \item Note that the oracles $\textbf{Oracle}
_{\alpha,n, \epsilon}$
for various values of $\text{\ensuremath{\alpha} are invoked at most \ensuremath{2^{\tilde{O}\parr{\frac{\sqrt{n}}{\epsilon}}}}}$
times. Therefore, \Cref{lem: main lemma-1} tells us that for each of this invocations the algorithm $\textbf{Oracle}
_{\alpha,n, \epsilon}$ conforms to its specification with probability at least $1-\frac{1}{2^{5n}}$. Via union bound we see that event (1) holds with probability at least\footnote{We assume that $\epsilon$ is such that $2^{0.1 n}$ exceeds the number $2^{\tilde{O}\parr{\frac{\sqrt{n}}{\epsilon}}}$ of times that $\textbf{Oracle}
_{\alpha,n, \epsilon}$ is invoked (for different values of $\alpha$. Otherwise, the run-time budget is sufficient to store entire truth-tables of functions over $\bits^n$ and 
statement in \Cref{alg:cube-corrector} is achieved by the trivial algorithm that uses a linear program to fit the best montone real-valued function and then rounds it to be $\bits$-valued. See \Cref{subsec: very small epsilon} for further details.}  $1-\frac{1}{2^{4n}}$.
\item Event (2) holds with probability at least $1-\frac{1}{2^{4n}}$ via \Cref{thm:l1-cube-corrector}.
\item Event (3) holds with probability at least $1-\frac{1}{2^{4n}}$ via \Cref{fact: rounding}
\end{itemize} Via union bound, we see that with probability at least $1-\frac{1}{2^{3n}}$ all these events hold, which we will assume for the rest of the proof.

Recall that $\text{opt}$ stands for the distance of $f$ to the closest monotone
function. We first claim that the algorithm will break out of the
loop in \Cref{line:return} for some value $\alpha^{*}\leq 2\text{opt}+150\epsilon$,
which we argue as follows:
 If $\alpha^{*}>2\text{opt}+150\epsilon$, then for some\footnote{Note that $\text{opt}\leq1/2$, because the function $f$ is at least
$1/2$-close to either the all-ones or all-zeroes functions, which
are both monotone. Therefore some value of $\alpha$ in the range
$\left[2\text{opt+100\ensuremath{\epsilon},}2\text{opt+150\ensuremath{\epsilon}}\right]$
is necessarily considered by the algorithm as it is trying all values
$\alpha=\epsilon,2\epsilon,3\epsilon,\cdots1-\epsilon,1+200\epsilon$.} $\alpha\in\left[2\text{opt+100\ensuremath{\epsilon},}2\text{opt+150\ensuremath{\epsilon}}\right]$
the ellipsoid algorithm failed to find some polynomial $P$ on which
$\textbf{Oracle} _{\alpha,n, \epsilon}$ returns ``Yes''. We claim that this
is impossible. Indeed, let $\mathcal{C}_{\text{convex}}$ be the set consisting of
degree-$\left\lceil \frac{4 \cdot \sqrt{n}}{\eps} \log \frac{4}{\eps} \right\rceil$ polynomials $P'$ with
$\norm{P'}_{2}\leq1$ that satisfies the following two conditions:
\begin{itemize}
\item $P'$ is $\epsilon$-close in $L_{1}$ distance to some monotone
function $f_{\text{monotone}}:\left\{ \pm1\right\} ^{n}\righ\left[-1,1\right]$,
and 
\item $P'$ is $\parr{\alpha+\epsilon}$-close in $L_{1}$ distance
to the function $f$ which we are trying to learn.
\end{itemize}
We make the following
observations:
\begin{itemize}
\item The set $\mathcal{C}_{\text{convex}}$ is a convex set, because (a) the set of all monotone
functions $f_{\text{monotone}}:\left\{ \pm1\right\} ^{n}\righ\left[-1,1\right]$
is convex, (b) the set of points $\parr{\alpha+\epsilon}$-close in
$L_{1}$ distance to some specific convex set is itself convex, and (c)
the intersection of two convex sets is a convex set (in this case one
convex set is the set functions $\left\{ \pm1\right\} ^{n}\righ\left[-1,1\right]$
that are $\parr{\alpha+\epsilon}$-close in $L_{1}$ distance a monotone
functions and the other convex set is is the set of all degree-$\left\lceil \frac{4 \cdot \sqrt{n}}{\eps} \log \frac{4}{\eps} \right\rceil$
polynomials with with $\norm{P'}_{2}\leq1$). 
\item The set $\mathcal{C}_{\text{convex}}$ contains an $L_{2}$ ball of radius at least $\epsilon \cdot n^{-\frac{1}{2}\left\lceil \frac{4 \cdot \sqrt{n}}{\eps} \log \frac{4}{\eps} \right\rceil}$.
In other words, in $\mathcal{C}_{\text{convex}}$ there is some degree\-$\left\lceil \frac{4 \cdot \sqrt{n}}{\eps} \log \frac{4}{\eps} \right\rceil$
polynomial $P_{0}$ such that any degree-$\left\lceil \frac{4 \cdot \sqrt{n}}{\eps} \log \frac{4}{\eps} \right\rceil$
polynomial $P'$ that is $\epsilon$-close to $P_{0}$ in $L_2$ norm is also in
$\mathcal{C}_{\text{convex}}$. Let $f_{\text{monotone, optimal}}:\left\{ \pm1\right\} ^{n}\righ\left\{ \pm1\right\} $
be the monotone function for which it is the case that $\Pr_{\vect x\sim\left\{ \pm1\right\} ^{n}}\left[f_{\text{monotone, optimal}}(\vect x)\neq f(\vect x)\right]=\text{opt}$,
and let $P_{0}$ be a degree-$\left\lceil \frac{4 \cdot \sqrt{n}}{\eps} \log \frac{4}{\eps} \right\rceil$
polynomial that is $\epsilon$-close to $f_{\text{monotone, optimal}}$
in $L_{1}$ norm (such polynomial has to exist by \Cref{fact:A monotone Boolean functions well-approximated by polys}).
Then, $P_{0}$ is $\left(2\text{opt}+\epsilon\right)$-close to $f$
in $L_{1}$ norm and $\epsilon$-close to monotone in $L_{1}$ norm. In other words, the set $\mathcal{C}_{\text{convex}}$ contains an $L_{1}$-ball of radius $\epsilon$. Via the standard inequality between the $L_{1}$ and $L_{2}$ norms, in $d$ dimensions every $L_1$ ball or radius $\epsilon$ contains an $L_2$ ball of radius at most $\epsilon/\sqrt{d}$. Our claim follows, since the space of degree-$\left\lceil \frac{4 \cdot \sqrt{n}}{\eps} \log \frac{4}{\eps} \right\rceil$ over $\R^d$ has dimension at most $n^{\left\lceil \frac{4 \cdot \sqrt{n}}{\eps} \log \frac{4}{\eps} \right\rceil}$.
\item Since the procedure Oracle$_{\alpha,n, \epsilon}$ is assumed to satisfy the
specifications given in \Cref{lem: main lemma-1} and for this specific
value of $\alpha$ it never gave the response ``yes'', then for
every query $P$ to Oracle$_{\alpha,n, \epsilon}$, the oracle returned some halfspace
that separates $P$ from the convex set $\mathcal{C}_{\text{convex}}$.
\end{itemize}
From \Cref{fact: ellipsoid algorithm} we know that under these conditions
the ellipsoid algorithm will necessarily in time \\ $\text{\text{poly}}\left(n^{\left\lceil \frac{4 \cdot \sqrt{n}}{\eps} \log \frac{4}{\eps} \right\rceil},\log\parr{R/r}\right)=n^{O\parr{\left\lceil \frac{4 \cdot \sqrt{n}}{\eps} \log \frac{4}{\eps} \right\rceil}}$
find some polynomial $P$ that is in $\mathcal{C}_{\text{convex}}$. For this particular
polynomial, the specifications in \Cref{lem: main lemma-1} require
the oracle Oracle$_{\alpha,n, \epsilon}$ to give a response ``yes'', which
gives us a contradiction.
Thus, the function $P^\text{GOOD}_{\text{TRIMMED}}$ will be $O(\epsilon)$-close
to monotone in $L_{1}$ norm and will satisfy 
\[
\norm{P^\text{GOOD}_{\text{TRIMMED}}-f}_1 \leq2\text{opt}+O(\epsilon).
\]
Combining this with \Cref{eq: Jane's corrector works} yields 
\begin{multline*}
\norm{P^\text{GOOD}_{\text{CORRECTED}}-f}_1 
\leq2\text{opt}+O(\epsilon)+\norm{P^\text{GOOD}_{\text{TRIMMED}}-P^\text{GOOD}_{\text{CORRECTED}}}_1=2\text{opt}+O(\epsilon).
\end{multline*}
We know that $\norm{P^\text{GOOD}_{\text{TRIMMED}}-P^\text{GOOD}_{\text{CORRECTED}}}_1 \leq O(\epsilon)$  because $P^\text{GOOD}_{\text{TRIMMED}}$ is $O(\epsilon)$-close to monotone by \Cref{eq: Jane's corrector works}.
Now, combining the inequality above with Equation \ref{eq: rounding works right}
gives us 
\[
\Pr_{\vect x\sim\left\{ \pm1\right\} ^{n}}\left[\sign(P^\text{GOOD}_{\text{CORRECTED}}(\vect x)-t^{*})\neq f\right]\leq\frac{1}{2}\norm{P^\text{GOOD}_{\text{CORRECTED}}-f}_1+\epsilon\leq\text{opt}+O(\epsilon).
\]
Finally, we see that since the function $P^\text{GOOD}_{\text{CORRECTED}}\left\{ \pm1\right\} ^{n}\righ\left[-1,+1\right]$
is monotone we have that the $\left\{ \pm1\right\} $-valued function
$\sign(P^\text{GOOD}_{\text{CORRECTED}}(\vect x)-t^{*})$
is also monotone, which finishes our argument.

%% file: appendix.tex





\section{Rounding of real-valued functions to Boolean.}
\begin{fact}
\label{fact: rounding}Suppose we have two functions $g:\left\{ \pm1\right\} ^{n}\righ\R$
and $f:\left\{ \pm1\right\} ^{n}\righ\left\{ \pm1\right\} $. Let $T$ be a set of at least $\frac{40}{\epsilon^2} \log\parr{\frac{20}{\epsilon \delta} \log{\frac{1}{\delta}}}$ i.i.d. uniformly random elements of $\bits^n$, and let $\text{ThresholdCandidates}\subset [-1,1]$ be a set of $\frac{20}{\epsilon} \log{\frac{1}{\delta}}$ i.i.d. uniformly random elements of $[-1,1]$. Let 
\[
t^*
:=
\argmin_{t\in \text{ThresholdCandidates}}
\frac{1}{|T|}
\sum_{\vect x \in T} \abs{\sign(g(\vect x)-t) - f(\vect x)}
\]
 Then, with probability at least $1-\delta$ it is the case that
\[
\Pr_{\vect x\sim\left\{ \pm1\right\} ^{n}}\left[\sign(g(\vect x)-t^*)\neq f\right]\leq \lfrac{1}{2}\norm{f - g}_1+\epsilon
\]

\end{fact}

\begin{proof}
We get that 
\[
\E_{t\sim\left[-1,1\right]}\pars{\E_{\vect x\sim\left\{ \pm1\right\} ^{n}}\left[\abs{\sign(g(\vect x)-t)-f(\vect{x})}\right]}\leq\norm{f - g}_1
\]
directly via linearity of expectation. Now, the random variable
$\E_{\vect x\sim\left\{ \pm1\right\} ^{n}}\left[\abs{\sign(g(\vect x)-t)-f(\vect{x})}\right]$
(with randomness taken over $t$)
is always in $[0,2]$ and has some expectation $E\in[0,2]$ which
is at most $\norm{f - g}_1$.
By Markov's inequality, we have
\[
\Pr_{t\sim\left[-1,1\right]}\left[\E_{\vect x\sim\left\{ \pm1\right\} ^{n}}\left[\abs{\sign(g(\vect x)-t)-g(\vect{x})}\right]\geq E+\epsilon/2\right]\leq\frac{E}{E+\epsilon/2}\leq\frac{2}{2+\epsilon/2}\leq1-\frac{\epsilon}{4}.
\]
Since the set $\text{ThresholdCandidates}$ consists of $\frac{20}{\epsilon}\log \frac{1}{\delta}$ i.i.d. uniform elements in $[-1,1]$, then with probability $1-\delta$ or more, some $t$ in $\text{ThresholdCandidates}$ will satisfy the condition that $\E_{\vect x\sim\left\{ \pm1\right\} ^{n}}\left[\abs{\sign(g(\vect x)-t)-g(\vect{x})}\right]$ is in $[0,E+\epsilon/2]$.

Finally, from the Hoeffding bound and union bound we observe that with probability at least $1-\frac{\delta}{2}$ it is the case that
\[
\max_{t \in \text{ThresholdCandidates}} \abs{\frac{1}{|T|}
\sum_{\vect x \in T} \abs{\sign(g(\vect x)-t) - f(\vect x)}
-
E_{\vect x \sim \bits^n} \abs{\sign(g(\vect x)-t) - f(\vect x)}
}
\leq \frac{\epsilon}{4}.
\]
Overall, we see that with probability at least $1-\delta$ it is the case that
\[
\Pr_{\vect x\sim\left\{ \pm1\right\} ^{n}}\left[\sign(g(\vect x)-t^*)\neq f\right]
\leq
\frac{1}{|T|}
\sum_{\vect x \in T} \abs{\sign(g(\vect x)-t^*) - f(\vect x)}
+
\frac{\epsilon}{4}
\leq
\lfrac{1}{2}\norm{f - g}_1+\epsilon
\]
This finishes the proof.
\end{proof}

\section{Agnostic learning algorithms handling randomized labels.}
\label{section: learning with randomized labels}
It is customary in the agnostic learning literature to consider a setting that is slightly more general than the one in \Cref{theorem: main theorem}. 
Specifically, one is given pairs of i.i.d. elements $\{(x_i,y_i)\}$ from a distribution $D_{\text{pairs}}$, where the distribution of each $x_i$ by itself is uniform. The aim here is to output an efficiently-evaluable succinct representation of a function $g$ for which 
\begin{equation}
\label{eq: agnostic learning guarantee with randomized labels}
\Pr_{(\bx,\by)\sim D_{\text{pairs}}}[g(\bx)\neq \by]
\leq
\min_{\text{monotone $f_{\text{mon}}:\bits^n\rightarrow\bits$}}
\Pr_{(\bx,\by)\sim D_{\text{pairs}}}[f_{\text{mon}}(\bx)\neq \by]
+O(\epsilon).
\end{equation}
The only difference between this setting and the one in \Cref{theorem: main theorem} is that here the label $y$ doesn't have to be a function of example $x$; it is possible to receive the same example $x$ twice accompanied by different labels. 
Here we argue that \Cref{theorem: main theorem} extends directly into this slightly more general setting. Formally, we show that 

\begin{theorem}
\label{theorem: main theorem randomized labels} For all sufficiently large integers $n$ the following holds.
There is an algorithm that runs in time
$2^{\tilde{O}\parr{\frac{\sqrt{n}}{\epsilon}}}$ and given i.i.d.
samples of pairs $\{(x_i,y_i)\}$ from a distribution $D_{\text{pairs}}$, where the marginal distribution over $x$ is uniform, does the following.
With probability at least $1-\frac{1}{2^{0.5n}}$ the algorithm outputs a representation of a monotone function $g:\left\{ \pm1\right\} ^{n}\righ\left\{ \pm1\right\} $ of size $2^{\tilde{O}\parr{\frac{\sqrt{n}}{\epsilon}}}$
that satisfies \Cref{eq: agnostic learning guarantee with randomized labels}.
\end{theorem}
\subsection{Case 1: $\epsilon$ is very small.}
\label{subsec: very small epsilon}
We will consider two cases. First of all, suppose $\epsilon$ is so small that the run-time of the algorithm in \Cref{theorem: main theorem} exceeds $2^{0.1n}$. In this case, the following algorithm runs in time $\poly(2^n, 1/\epsilon)$ and outputs and efficiently-evaluable succinct representation of a function $g$ for which \Cref{eq: agnostic learning guarantee with randomized labels} holds:
\begin{enumerate}
    \item Draw two sets $T_1$ and $T_2$, each of $100 n^5 \cdot 2^n/\epsilon^2$ example-label pairs from $D_{\text{pairs}}$.
    \item For each $x \in \bits^n$ let $h(x)$ be $\frac{1}{\abs{(x_i,y_i)\in T_1 ~\text{s.t.}:~  x_i = x}}\sum_{(x_i,y_i)\in T_1 \text{ s.t.} ~  x_i = x} y_i $.
    \item Via a size-$2^{O(n)}$ linear program, find the monotone function $q:\bits^n \righ [-1,1]$ that is closest to $h$ is $\ell_1$ distance.
    \item Output the function $g$ defined so $g(x):=\sign(q(x)-t^*)$, where $t^*$ is obtained as in \Cref{fact: rounding} using the samples in $T_2$.
\end{enumerate}
The function $g$ we output above with high probability satisfies \Cref{theorem: main theorem} for the following reason. First of all, via the standard coupon-collector argument with probability at least $1-\frac{1}{2^{5n}}$ for every $x \in \bits^n$ there will be at least $10^2/\epsilon^2$ elements in $(x_i, y_i)$ in $T$ for which $x_i=x$. Using the Hoeffding bound and the union bound, we see that with probability at least $1-\frac{1}{2^{2n}}$ we have 

\begin{equation}
    \label{eq: distance to h good proxy for error}
\left\lvert h(x) -\E_{(\bx', \by') \sim D_{\text{pairs}}}\pars{
\by' \bigg\vert \bx'=x
} \right\rvert \leq \frac{\epsilon}{2}.
\end{equation}
Now, from steps (3) and (4) we have
\begin{equation}
\label{eq: approximately best monotone approximator to h}
\frac{\norm{h-g}_1}{2}
\leq
\lfrac{1}{2} \dist_1(h, \mathrm{mono})
+\epsilon.
\end{equation}
Therefore, we can combine \Cref{eq: distance to h good proxy for error} and \Cref{eq: approximately best monotone approximator to h} to obtain 
\begin{equation}
\Pr_{(\bx,\by)\sim D_{\text{pairs}}}[g(\bx)\neq \by]
\leq
\min_{\text{monotone $f_{\text{mon}}:\bits^n\rightarrow\bits$}}
\Pr_{(\bx,\by)\sim D_{\text{pairs}}}[f_{\text{mon}}(\bx)\neq \by]
+O(\epsilon),
\end{equation}
which finishes the proof for this case.
\subsection{Case 2: $\epsilon$ is not too small.}

Now, we proceed to the other case when $\epsilon$ is not too small and the algorithm in \Cref{theorem: main theorem} runs in time at most $2^{0.1n}$ (and therefore uses at most $2^{0.1n}$ samples). In this case, we claim that simply running the algorithm in \Cref{theorem: main theorem} will give an efficiently evaluable succinct description of a function $g$ that satisfies the guarantee in \Cref{eq: agnostic learning guarantee with randomized labels}. 

We now proceed to show that the guarantee in \Cref{eq: agnostic learning guarantee with randomized labels} will indeed be achieved.
Define a random function $f_{\text{random}}: \bits^n \righ \bits$, so for all $x\in \bits^n$ the value $f_{\text{random}}(x)$ is chosen independently such that $f_{\text{random}}(x)=1$ with probability $\Pr_{(\bx', \by') \sim D_{\text{pairs}}}\pars{
\by'=1 ~\vert~ \bx'=x
}$ and $f_{\text{random}}(x)=-1$ with probability $\Pr_{(\bx', \by') \sim D_{\text{pairs}}}\pars{
\by'=-1 ~\vert~ \bx'=x
}$.  Consider the following two scenarios:
\begin{itemize}
    \item \textbf{Scenario I:} The samples $\{(\bx_i, \by_i)\}$ given to the algorithm from \Cref{theorem: main theorem} are indeed i.i.d. samples coming from $D_{\text{pairs}}$.
    \item \textbf{Scenario II:} The samples $\{(\bx_i, \by_i)\}$ given to the algorithm from \Cref{theorem: main theorem} are sampled as follows: (i) $\bx_i$ are i.i.d. uniform from $\bits^n$ (ii) $\by_i=f_{\text{random}}(\bx_i)$.
\end{itemize}
First we argue that in Scenario II with probability at least $1-\frac{2}{2^n}$ the function $g$ given by the algorithm from \Cref{theorem: main theorem} satisfies \Cref{eq: agnostic learning guarantee with randomized labels}, (here the probability is over the choice of $f_{\text{random}}$, choice of the samples, and the randomness of the algorithm itself). Indeed, let $f_{\text{mon}}^*$ be the function that minimizes the right side of \Cref{eq: agnostic learning guarantee with randomized labels}. From the Hoeffding's bound, it follows that with probability at least\footnote{Here we used that $\epsilon\geq \frac{1}{\sqrt{n} \poly \log n}$, because otherwise $\epsilon$ would be too small and we would be in the other case when the run-time of the algorithm in \Cref{theorem: main theorem} exceeds $2^{0.1n}$. Also, we note that a much stronger bound can be deduced from the Hoeffding bound, but we only need a bound of $1-\frac{1}{2^n}$.} $1-\frac{1}{2^n}$ over the choice of $f_{\text{random}}$  it is the case that 
\begin{equation}
\label{eq: fixing function randomly does not change opt much}
\abs{
    \underset{\bx \sim \bits^n}{\Pr}\pars{f_{\text{random}}(\bx) \neq f_{\text{mon}}^*(\bx)}
    - \underset{(\bx,\by)\sim D_{\text{pairs}}}{\Pr}[f_{\text{mon}}^*(\bx)\neq \by]
    }\leq \epsilon.
\end{equation}
Now, \Cref{theorem: main theorem} implies that with probability at least $1-\frac{1}{2^n}$
\begin{multline}
\label{eqm: applying theorem for fixed labels}
\underset{\bx\sim \bits^n}{\Pr}[g(\bx)\neq f_{\text{random}}(\bx)]
\leq
\dist_0(f_{\text{random}}, \mathrm{mono})
+O(\epsilon)\leq
\underset{\bx\sim \bits^n}{\Pr}[f_{\text{mon}}^*(\bx)\neq f_{\text{random}}(\bx)]
+O(\epsilon)
.
\end{multline}
Combining Equations \ref{eq: fixing function randomly does not change opt much} and \ref{eqm: applying theorem for fixed labels} we we see that with probability at least $1-\frac{2}{2^n}$, the function $g$ given by the algorithm from \Cref{theorem: main theorem} satisfies \Cref{eq: agnostic learning guarantee with randomized labels} in Scenario II. 

Finally, we argue that \Cref{eq: agnostic learning guarantee with randomized labels} will be satisfied also in Scenario I with probability at least $1- \frac{1}{2^{0.5n}}$ for sufficiently large $n$. 
Conditioned on the absence of sample pairs $(\bx_i, \by_i)$ and $(\bx_j, \by_j)$ with $\bx_i=\bx_j$, the distributions over samples in Scenario I and Scenario II are the same,
Hence it suffices to argue that the collision probability is low, given that the value of $\epsilon$ is such that the algorithm from \Cref{theorem: main theorem} uses at most $2^{0.1 n}$ samples.
By taking a union bound over all pairs of samples, we bound the probability of such collision by $\frac{2^{0.2n}}{2^n}=2^{-0.8n}$. Thus, information-theoretically, any algorithm can distinguish between Scenario I and Scenario II with an advantage of only at most $2^{-0.8n}$. In particular, this is true of the algorithm that checks whether  \Cref{eq: agnostic learning guarantee with randomized labels} applies. Thus, indeed  \Cref{eq: agnostic learning guarantee with randomized labels} will be satisfied also in Scenario I with probability at least $1-\frac{2}{2^n}-\frac{1}{2^{0.8n}}\geq 1- \frac{1}{2^{0.5n}}$, which finishes the proof of \Cref{theorem: main theorem randomized labels}.

\section{Proofs deferred from \Cref{sec:corrector}}
\label{apx:corrector}

\begin{proof}[Proof of \Cref{lem:bitwise}]
Let $x$ and $y$ be comparable elements of $P$; w.l.o.g. $x \prec_P y$. 
It is sufficient to show that $f(x) > f(y)$ if and only if there is some $i$ for which $x \prec_{P_i} y$ and $f_i(x) > f_i(y)$. 
We claim that this $i$ is the most significant bit in which $f(x)$ and $f(y)$ differ. 
It is certainly true that $f(x) > f(y)$ if and only if $f_i(x) > f_i(y)$ for this $i$,
and since $f_j(x) = f_j(y)$ for all $j < i$ by the choice of $i$, 
we have $x \prec_{P_i}y$ as well.  
\end{proof}

\begin{proof}[Proof of \Cref{lem:l1-error-preservation}]
Since $m$ is monotone, certainly $m(x) \le m(y)$, and since $f$ violates monotonicity on this pair, certainly $f(x) \ge f(y)$ (and therefore $g(y) \ge g(x)$).
We will examine the contribution of $x$ and $y$ to each of $||f-m||_1$ and $||g-m||_1$. 
We have the following cases: 
\begin{itemize}
		\item $f(y) \le f(x) \le  m(x) \le m(y)$: then 
				\begin{align*}
						|m(x) - f(x)| + |m(y) - f(y)| &= m(x) + m(y) - (f(x) + f(y)) \\
						= &m(x) + m(y) - (g(x) + g(y)) \\
						= &|m(x) - g(x)| + |m(y) - g(y)|.
				\end{align*}
				The distance of this pair does not change. The case of $m(x) \le m(y) \le f(x) \le f(y)$ is symmetric. 
		\item $f(y) \le m(x) \le m(y) \le f(x)$: then   
				\begin{align*}
						|m(x) - f(x)| + |m(y) - f(y)| &= (f(x) - m(x)) + (m(y) - f(y))  \\
													  &\ge (f(x) - m(y)) + (m(x) - f(y)) \\
													  &= |g(y) - m(y)| + |g(x) - m(x)|.
				\end{align*}
				The distance of this pair does not increase. The case of $m(x) \le f(y) \le f(x) \le m(y)$ is symmetric.
		\item $f(y) \le m(x) \le f(x) \le m(y)$: then   
				\begin{align*}
						|m(x) - f(x)| + |m(y) - f(y)| &= (f(x) - m(x)) + (m(y) - f(y))  \\
													  &\ge (m(x) - f(y)) + (m(y) - f(x)) \\
													  &= |g(x) - m(x)| + |g(y) - m(y)|.
				\end{align*}
				The distance of this pair does not increase. The case of $m(x) \le f(y) \le m(y) \le f(x)$ is symmetric.
\end{itemize}
\end{proof}

%

\begin{proof}[Proof of \Cref{thm:l1-cube-corrector}]
Let $f:\bits^n \to [-1,1]$ be $\alpha$-close to monotone in $\ell_1$ distance. 
We call the algorithm {\sc HypercubeCorrector}$(f, \eps, r)$ with a random seed $r$ of length $2^{O(\sqrt{n \log (1/\eps)} \log n )}$. 
First we set the poset to be the truncated cube of width $\sqrt{2n \log 2/\eps}$, 
which is a poset such that every element has at most $2^{O(\sqrt{n \log (1/\eps)} \log n )}$ predecessors and successors.
The representation of this poset (not its transitive closure) has size $\poly(n, \log(1/\eps))$.
Then we set $f'$ to be a function that discretizes $f$ to $2/\eps$ possible values. 
This representation has size $O(s_f/\eps)$. 
Then we set $f''$ to be a function that computes the Hamming weight of $x$, 
then either calls $k$-{\sc Corrector} or outputs a constant.
So its size is the size of the $k$-{\sc Corrector} representation times some overhead that is polynomial in $n$ and $1/\eps$. 
Since the $\Delta$ parameter for the truncated cube is $2^{O(\sqrt{n \log (1/\eps)} \log n )}$, the $h$ parameter is $O(\sqrt{n})$, and the $N$ parameter is $< 2^n$,
the worst-case running time and query complexity of this instance of $k$-{\sc Corrector} is $2^{O(\sqrt{n} \log n \log^{3/2} (1/\eps))}$ by \Cref{lem:k-correctness}.
Thus the representation size of the $k$-{\sc Corrector} instance is $2^{\tilde{O}(\sqrt{n} \log^{3/2} (1/\eps))}$, and so the representation size of $f''$ is $2^{\tilde{O}(\sqrt{n} \log^{3/2} (1/\eps))} \cdot s_f$. 
With the random seed of length $2^{O(\sqrt{n \log (1/\eps)} \log n )}= \poly(\Delta \log N)$, $k$-\textsc{Corrector} succeeds with probability $N^{-10} \le 2^{-10n}$.

\end{proof}

\section{Proofs deferred from \Cref{sec:matching}}
\label{apx:matching}

\begin{proof}[Proof of \Cref{lem:matching-approx-dist}]
The proof of $\dist_1(f, \mathrm{mono}) \ge W/N$ is straightforward; for any edge $(x,y)$, $x \prec y$ in the matching, any monotone function must have $g(y) \ge g(x)$ and thus $(f(x) - g(x)) + (g(y) - f(y)) \ge f(x) - f(y)$. So the contribution of $x$ and $y$ to the $\ell_1$ distance is at least the weight of $(x,y)$.

For the other direction, we give a proof exactly analogous to the max-weight matching characterization of distance to the class of Lipschitz functions, presented in \cite{berman_l_p-testing_2014}. 
Let $g$ be the closest monotone function to $f$ in $\ell_1$-distance. 
We will partition the vertices of the cube into three classes: $V_> := \{x~|~ f(x) > g(x)\}$, $V_< := \{x~|~ f(x) < g(x)\}$, and $V_= := \{x~|~ f(x) = g(x)\}$. 
We will duplicate the vertices of $V_=$ and group one copy with $V_>$ and one copy with $V_<$,
to form vertex sets $V_\ge$ and $V_\le$.  
The duplicated copies of $x$ will be denoted $x_\ge$ and $x_\le$. 
We define the bipartite graph $B_{f,g}$ to be the graph on $V_\ge \times V_\le$ with an edge $(x,y)$ if $x \prec y$ and $g(x) = g(y)$. 
The weight of the edge $(x,y)$ is the same as it is in $\viol(f)$; it is just $f(x) - f(y)$. 
Intuitively, a matching in $B_{f,g}$ will represent a set of edges along which some a minimal amount of label mass is transferred to correct monotonicity.
First, we claim that $B_{f,g}$ has a matching which matches every vertex in $V_> \cup V_<$. 
This will follow from Hall's marriage theorem if we can show that for every $A \subseteq V_>$ or $A \subseteq V_<$, we have $|A| \le |N(A)|$. 

Suppose for contradiction that the marriage condition is false, and without loss of generality let $A$ be the largest subset of $V_>$ for which $|A| > |N(A)|$. We would like to claim that for any $x \in A \cup N(A)$ and $y \not\in A \cup N(A)$, if $x \prec y$ then $g(x) < g(y)$.
We consider four possible cases: 
\begin{itemize}
    \item[a)] If $x \in A$, $y \in V_>$, $x \prec y$, and $g(x) = g(y)$, then $y \in A$ as well, by the choice of $A$ to be the largest set that fails the marriage condition. This is because $N(y) \subseteq N(x)$: any neighbor $z$ of $y$ must have $g(z) = g(y) = g(x)$, have $x \prec y \prec z$, and be in $V_\le$, which makes it a neighbor of $x$.
    \item[b)] If $x \in N(A)$, $y \in V_\le$, $x \prec y$, and $g(x) = g(y)$, then $g(y) = g(x) = g(z)$ and $z \prec x \prec y$ for some $z \in A$, so $y \in N(A)$. 
    \item[c)] If $x \in A$, $y \in V_\le$, $x \prec y$, and $g(x) = g(y)$, then $y \in N(A)$. 
    \item[d)] If  $x \in N(A)$, $y \in V_>$, $x \prec y$, and $g(x) = g(y)$, then $g(y) = g(x) = g(z)$ and $z \prec x \prec y$ for some $z \in A$, so as in case (a) we have $N(y) \subseteq N(z)$ and therefore $y \in A$. 
\end{itemize}

We have shown that for any $x \in A \cup N(A)$ and $y \not\in A \cup N(A)$, if $x \prec y$ then $g(x) < g(y)$.
Then there is some $\delta > 0$ for which $g(x)$ can be increased by $\delta$ for every $x \in A \cup N(A)$ without breaking monotonicity.
This decreases $||f - g||_1$ by $\delta(|A| - N(A)|) >0$, which contradicts the assumption that $g$ is the closest monotone function.

Having proven that $B_{f,g}$ contains a matching $M'$ on all vertices in $V_> \cup V_<$, 
we will now show that its weight is equal to $N||f - g||_1$,
using the fact that $g(x) = g(y)$ for all $(x,y) \in M'$: 
\[\sum_{(x,y) \in M'} f(x) - f(y) = \sum_{(x,y) \in M'} f(x) - g(x) + g(y) - f(y) = \sum_{x \in V_> \cup V_<} |f(x) - g(x)| = N ||f - g||_1.\]

We will now find a matching $M$ in $\viol(f)$ of equal weight. 
First replace each $x_\le$ and $x_\ge$ with $x$, obtaining an edge set in $\viol(f)$ of equal weight that is not necessarily a matching, but is a set of disjoint paths. 
We replace each path with the edge between its endpoints; i.e.
if there is some pair of edges $(y, x_\le)$ and $(x_\ge, z)$, then we know that $y \prec x \prec z$ and $f(y) - f(z) = ((f(y) - f(x) + (f(x) - f(z))$, so the matching edge $(y,z)$ has weight equal to the total weight of the path it replaces.
Then $M$ is a matching in $\viol(f)$ of weight equal to $N||f- g||_1$, which is equal to $N \cdot \dist_1(f, \mathrm{mono})$.
\end{proof}

\begin{proof}[Proof of \Cref{lem:global matching}]
Fix the random seed $r$ and assume all calls to the algorithm of \cite{ghaffari_local_2022} using $r$ succeed.
Let $M'$ be a maximum-weight matching over $\viol(f)$, and let $M$ be a matching returned by \textsc{MatchViolations}. 
We will use $M$ to refer to the matching and its succinct representation interchangeably.
For each edge $e \in M'$, let $w_e$ be the weight of $e$ (i.e. the violation score of its endpoints), 
and $\delta_e$ be the total weight of edges in $M \setminus M'$ that share an endpoint with $e$. 

First we show by induction that at the start of each iteration $i$,
$M$ is maximal over the subgraph of $TC(P)$ induced by edges of weight greater than $2^{-(i-1)}$. 
In the base case, $M$ is initialized to be the empty matching, which is maximal on the edges of weight $> 2$, as there are no such edges.
In the inductive case, we assume the invariant is still true at the start of iteration $i$. 
Then when \textsc{FilterEdges} (\Cref{alg:filter edges}) is called in iteration $i+1$, the vertices removed are exactly those that are either already in $M$, or not incident to any edges of weight greater than $t = 2^{-i}$.
Then by the maximality of the matching computed by \textsc{GhaffariMatching} on the filtered subgraph, any edge not in that matching must satisfy one of the following criteria:
\begin{itemize}
\item it has weight at most $2^{-i}$,
\item it has an endpoint in $M$,
\item it shares an endpoint with another edge in \textsc{GhaffariMatching}.
\end{itemize}

So after the new edges of in \textsc{GhaffariMatching} are added to $M$, $M$ is maximal over the $2^{-i}$-heavy edges as desired.

Now we claim that $\delta_e \ge  w_e/2$ for any edge $e \in M' \setminus M$ of weight at least $\eps$. 
This is because after the first round for which $t < w_e$, $M'$ must be maximal over the $t$-heavy edges. 
This $t$ is at least $w_e/2$, so if $e \not\in M$, then either it shares an endpoint with some edge of weight at least $w_e/2$ or its own weight is $\le \eps$. 
We then have 
\begin{align*}
		w(M') &= w(M \cap M') + \sum_{e \in M' \setminus M} w_e\\
		&\le w(M \cap M') + \sum_{e \in M' \setminus M} \max(2\delta_e, \eps) \\
		&\le w(M \cap M') + 2\sum_{e \in M' \setminus M} \delta_e + \eps N
\end{align*}
We claim that $\sum_{e \in M' \setminus M} \le 2 \cdot w(M \setminus M')$. 
This is because each edge in $M \setminus M$ shares an endpoint with at most 2 edges of $M' \setminus M$, otherwise $M'$ would not be a matching. 
Therefore, 
\begin{align*}
		w(M') &\le w(M \cap M') + 4\sum_{e \in M \setminus M'} w_e  + \eps N\\
		&\le 4 \cdot w(M) + \eps N 
\end{align*}
By \Cref{lem:matching-approx-dist}, $w(M') = N \cdot \dist_1(f, \mathrm{mono})$; therefore $w(M) \ge N( \lfrac{1}{4}\dist_1(f, \mathrm{mono}) - \eps)$ as desired. \\

We now bound the failure probability. When called with a random seed of length $\poly( \log N, \log \log (1/\eps))$ the algorithm of \cite{ghaffari_local_2022} can be made to succeed with probability $1 - (N^{-10}/\log(4/\eps))$. We use the random seed on at most $\log(4/\eps)$ different graphs, so by union bound, with probability $1 - N^{-10}$ all the calls succeed. 
By the same argument as in the proof of \Cref{thm:correction-main}, 
we may assume that $\log(1/\eps) \le N$,
and so the randomness complexity is $\poly(\Delta, \log N)$. 
\end{proof}

\section{Proof of \Cref{claim: big enough set of samples preserves L_1 distances}.}
\label{subsec: proof of claim big enough set of samples preserves L_1 distances}
Let us first recall the statement of the claim:
\concentration
First we bound the probability that the condition above holds for one specific $P$ with $\norm{P}_2\leq 1$. The condition $\norm P_{2}\leq1$ implies that $\max_{\vect x\in\left\{ \pm1\right\} ^{n}}\abs{P(\vect x)}\leq n^{d}$.
This implies, via the Hoeffding bound, that
\[
\Pr_{\text{choice of }T}\pars{\abs{\norm{f-P}_{1}-\E_{\vect x\sim T}\pars{\abs{f(\vect x)-P(\vect x)}}}>
\frac{\epsilon}{4}}
\leq\exp\parr{-\frac{\epsilon^{2}}{32}\frac{\abs T}{n^{2d}}}.
\]

We now move on to bounding the maximum over all degree-$\ensuremath{d}$
polynomials $P$ over $\left\{ \pm1\right\} ^{n}$ with $\norm P_{2}\leq1$.
We will need a collection $\mathcal{C}$ of degree $d$ polynomials
over $\left\{ \pm1\right\} ^{n}$, such that $\abs{\mathcal{C}}\leq\exp\parr{n^{d}\ln\frac{8n^{d}}{\epsilon}}$
so for every degree $d$ polynomial $P$ with $\norm P_{2}\leq1$
there is some element $P_{\text{closest}}\in\mathcal{C}$ for which
it is the case that 
\[
\max_{\vect x\in\left\{ \pm1\right\} ^{n}}\abs{P(\vect x)-P_{\text{closest}}(\vect x)}\leq\frac{\epsilon}{4}.
\]
Also, the $L_{2}$ norm of every element in $\mathcal{C}$ is at most
$1$. Such a set can be constructed by putting into $\mathcal{C}$ all polynomials
of the form $\sum_{\substack{S\subset[n]\\
\abs S\leq d
}
}c_{S}\parr{\chi_S(x)}$ with the coefficients $c_{S}$ taking values in $[-1,+1]$ rounded to the nearest multiple of $\frac{\epsilon}{8n^{d}}$, while discarding the polynomials whose
$L_{2}$ norm is larger than $1$. This way, since $\chi_S(x)\in\left\{ \pm1\right\} $,
when we round the coefficients of $P$ to a multiple of $\frac{\epsilon}{8n^{d}}$
the value at any $\vect x\in\left\{ \pm1\right\} ^{n}$ cannot change
by more than $\frac{\epsilon}{4}$, as there are at most $n^{d}$
contributing monomials \footnote{To have $\norm{P_{\text{closest}}(\vect x)}_{2}\leq\norm P_{2}\leq1$
we should round to the closest multiple of $\frac{\epsilon}{8n^{d}}$
that is smaller in the absolute value of the coefficient being rounded}.
The total number of such polynomials is at most $\parr{\frac{8n^{d}}{\epsilon}}^{n^{d}}=e^{n^{d}\ln\frac{8n^{d}}{\epsilon}}$.

Now, by taking a union bound on all elements of $\mathcal{C}$ we
get 
\[
\Pr_{\text{choice of }T}\pars{\max_{P\in\mathcal{C}}\abs{\norm{f-P}_{1}-\E_{\vect x\sim T}\pars{\abs{f(\vect x)-P(\vect x)}}}\leq\frac{\epsilon}{2}}\geq1-\exp\parr{-\frac{\epsilon^{2}}{32}\frac{\abs T}{n^{2d}}+n^{d}\ln\frac{8n^{d}}{\epsilon}}
\]
Finally, if the above holds, by choosing a polynomial $P_{\text{closest}}$ from $\cal{C}$ to minimize \\ $\max_{\vect x\in\left\{ \pm1\right\} ^{n}}\abs{P(\vect x)-P_{\text{closest}}(\vect x)}$ we get that 
\begin{multline*}
\Pr_{\text{choice of }T}\pars{\max_{\substack{\text{degree-\ensuremath{d} polynomial \ensuremath{P} over \ensuremath{\left\{  \pm1\right\}  ^{n}}}\\
\text{with \ensuremath{\norm P_{2}\leq}1}
}
}\abs{\norm{f-P}_{1}-\E_{\vect x\sim T}\pars{\abs{f(\vect x)-P(\vect x)}}}\leq\epsilon}\geq\\1-\exp\parr{-\frac{\epsilon^{2}}{8}\frac{\abs T}{n^{2d}}+n^{d}\ln\frac{4n^{d}}{\epsilon}}.
\end{multline*}
Substituting $\abs T\geq n^{5d}\frac{100}{\epsilon^{2}}\ln\frac{1}{\epsilon}\ln\frac{1}{\delta}$
we see that the above expression is at least $1-\delta$.